\documentclass[aps, prx, notitlepage, twocolumn,longbibliography,superscriptaddress,floatfix]{revtex4-2}
\usepackage{stmaryrd}
\usepackage{amssymb}
\usepackage{amsmath}
\usepackage{amsthm}
\usepackage{hyperref}
\usepackage{cleveref}
\usepackage{physics}
\usepackage{amsthm, thmtools}
\hypersetup{
    colorlinks,
    linkcolor={blue},
    citecolor={blue},
    urlcolor={blue},
    pdfauthor={Jon Nelson, Gregory Bentsen, Steven T. Flammia, Michael J. Gullans},
    pdftitle={Fault-tolerant quantum memory using low-depth random circuit codes}
}
\usepackage{listings}

\usepackage[utf8]{inputenc}
\usepackage{color}
\usepackage{qcircuit}
\usepackage{graphicx}
\usepackage{empheq}
\usepackage{placeins}
\usepackage{physics}
\usepackage{bbm}
\DeclareMathOperator*{\E}{\mathbf{E}}
\newcommand{\mc}{\mathcal}
\newcommand{\Id}{\mathbb{I}}

\DeclareMathOperator\eye{\mathbb{I}}

\declaretheorem[numberwithin=section]{theorem}

\declaretheorem[sibling=theorem]{fact}

\theoremstyle{definition}

\begin{document}
\title{Error correction phase transition in noisy random quantum circuits}

\author{Jon Nelson\textsuperscript{*}}

\author{Joel Rajakumar\textsuperscript{*}}

\author{Michael J. Gullans}
\affiliation{ Joint Center for Quantum Information \& Computer Science, University of Maryland and NIST}
\affiliation{Department of Computer Science,
	University of Maryland}
\thanks{These authors contributed equally to this work. \\email: nelson1@umd.edu, jrajakum@umd.edu}

\begin{abstract}
    In this work, we study the task of encoding logical information via a noisy quantum circuit. It is known that at superlogarithmic-depth, the output of any noisy circuit without reset gates or intermediate measurements becomes indistinguishable from the maximally mixed state, implying that all input information is destroyed. This raises the question of whether there is a low-depth regime where information is preserved as it is encoded into an error-correcting codespace by the circuit. When considering noisy random encoding circuits, our numerical simulations show that there is a sharp phase transition at a critical depth of order $p^{-1}$, where $p$ is the noise rate, such that below this depth threshold quantum information is preserved, whereas after this threshold it is lost. Furthermore, we rigorously prove that this is the best achievable trade-off between depth and noise rate for any noisy circuit encoding a constant rate of information. Thus, random circuits are optimal noisy encoders in this sense.
\end{abstract}

\maketitle

\section{Introduction} 

Random quantum circuits are known to exhibit a variety of phase transitions when subjected to intermediate nonunitary operations such as measurement or noise. The main feature that many of these effects have in common is that when the rate of these nonunitary operations exceeds a critical threshold, the circuit is no longer able to prevent quantum information from escaping. A prominent example is the measurement-induced phase transition in monitored random circuits \cite{PhysRevB.100.134306,PhysRevX.9.031009,PhysRevB.98.205136}, which has been directly linked to a transition in the ability of the circuit to retain memory of its input state \cite{Gullans_2020}. Similarly, Weinstein et al. \cite{Weinstein_2023} showed that randomly swapping qubits into the environment leads to a phase transition at which point the initial state's information becomes entirely lost to the environment. In a related study, Qian et al. \cite{qian2024coherentinformationphasetransition} identified a phase transition in coherent information driven by the interplay between single-qubit noise and the rate at which fresh ancilla qubits are introduced.

Most of these works focus on phase transitions that are independent of the circuit depth. In contrast, we study random circuits subject to mid-circuit noise, where the effects of noise compound with increasing depth. This cumulative behavior gives rise to a rich phase diagram that depends jointly on noise strength and circuit depth, which we fully characterize through numerical simulations and matching analytical bounds.

For sufficiently low depth, \cite{nelsonrandomcircuitcodes} designed a protocol that recovers logical information from the output of low-depth noisy random quantum circuits. However, it is known that at $\omega(\log n)$ \footnote{In this paper, we adopt the standard asymptotic notation 
$O$, $\Omega$, $\Theta$, $o$, and $\omega$ 
to describe relationships between the growth rates of two nonnegative functions 
$f(n)$ and $g(n)$ as $n \to \infty$. 
We write $f = O(g)$ when $\lim_{n \to \infty} f(n)/g(n) < \infty$, 
which is equivalent to stating that $g = \Omega(f)$. 
If $\lim_{n \to \infty} f(n)/g(n) = 0$, we denote this by 
$f = o(g)$, or equivalently $g = \omega(f)$. 
The notation $f = \Theta(g)$ indicates that both 
$f = O(g)$ and $f = \Omega(g)$ are satisfied. When a tilde is used, as in $\tilde{\Theta}$, 
it signifies that multiplicative polylogarithmic factors in $n$ 
are omitted from the scaling.} depth, the output of any noisy circuit converges to the maximally mixed state implying that all input information is destroyed \cite{M_ller_Hermes_2016,aharonov1996limitationsnoisyreversiblecomputation}. This suggests that as the depth is scaled up there must exist some transition between a regime where the input quantum state is recoverable to a regime where it is destroyed by the accumulation of noise. We numerically study this question for the case of circuits with gates drawn from either the Clifford or Haar ensemble and discover that a phase transition does indeed exist at a critical depth of $d^* = \Theta(1/p)$, where $p$ is the noise strength. Our numerics show that below this depth threshold, the majority of the logical information is still preserved by the circuit, but after this threshold the majority of this information is lost to the environment. This expands upon the scope of \cite{Weinstein_2023}, which shows a separate depth-independent critical noise rate that lies above ours, beyond which \textit{all} of the logical information is lost in finite depth.

Furthermore, we analytically prove that our observed above-threshold behavior holds for \textit{any} noisy circuit and not just the random ones we consider in our numerics. In particular, we prove that for $d> d^*$ the coherent information of the noisy circuit is upper bounded strictly below zero. Our numerics thus imply that random circuits are the best possible noisy encoders since they can preserve information all the way up to this depth upper bound \footnote{Optimal, here, is with respect to the asymptotic relationship between depth and noise strength and does not take into account other parameters such as encoding rate.}.

The optimality of random circuits for protecting information is not unprecedented. Random encoding circuits have shown great promise given their desirable coding parameters \cite{brown2013short,brown2015decoupling,gullans2021quantum, hayden2008decoupling}. For example, Hayden \textit{et al.} demonstrated that encoding with Haar random unitaries achieves the quantum channel capacity \cite{hayden2008decoupling}. Furthermore, Gullans \textit{et al.} numerically showed that even $O(\log n)$-depth random Clifford circuits are sufficient to saturate the capacity of the erasure channel~\cite{gullans2021quantum}. Intuitively, these random encoding circuits achieve such high performance due to their ability to scramble quantum information quickly \cite{brown2013scramblingspeedrandomquantum}. However, these results often assume that the logical information can be perfectly encoded before noise can act on the state. In reality, noise can act throughout the encoding circuit, which can destroy logical information before it is fully encoded. The challenge with this more realistic setting is that the fast-scrambling nature of random circuits not only encodes quantum information more quickly but also spreads errors faster throughout the circuit. At first glance, it is unclear which of these two competing effects wins out. Our results clarify this picture.

\section{Summary of Results}
In this work, we consider circuits that act on $n$ qubits where $k$ are considered to be logical qubits and the rest are ancilla qubits initialized to $\ket{0}^{n-k}$. The encoding rate is denoted by $r := k/n$. We first study brickwork circuits where 2-qubit Clifford gates are drawn randomly and applied to qubits on a 2D lattice. After each layer of gates, a round of depolarizing noise is applied. In order to characterize the ability to recover the original logical information, we then allow for a perfect syndrome measurement after the noisy circuit is applied. We numerically estimate the coherent information of this channel for various depths and noise rates, where the coherent information characterizes how much of the logical information is recoverable at the output of the channel and is a lower bound on the single-shot quantum channel capacity \cite{Lloyd_1997,devetak2004capacityquantumchannelsimultaneous}. Using an encoding rate of $r = 1/8$, our results show that when $p d$ is larger than approximately $0.6$, the coherent information of the channel approaches its minimum as the system size increases. On the other hand, when the product of depth and noise strength is below this threshold, then the coherent information instead approaches a finite positive value implying that the original logical information is partially recoverable.

Next, we consider a family of circuits where 2-qubit gates are chosen from the Haar measure and applied to disjoint random pairs of qubits in each layer. Then, a round of either depolarizing errors or amplitude damping errors are applied. As in the previous case, we are also interested in the coherent information of this channel as a function of noise strength and depth. However, Haar random circuits are not tractable to classically simulate so we must take a different approach to conduct the numerics. For this, we consider a standard approximation to the coherent information, which Gullans \textit{et al.} refer to as the log-average purity \cite{gullans2021quantum}. Although this is a heuristic approximation, it is expected that the fluctuations in the coherent information over the random circuits is small enough such that the log-average purity is a good approximation \cite{PhysRevX.7.031016,PhysRevX.8.021014,Zhou_2019}. This quantity can  be computed using standard calculations using the second moment operator of the Haar measure. Our numerics once again show a phase transition at a critical depth of $d^* = O(1/p)$.

Finally, we prove a rigorous upper bound on the coherent information for any noisy circuit (any gateset on any architecture). Informally, we prove the two following statements. First, we show that when $d>p^{-1}\log(1+1/r)$ then the coherent information, $I_c$, is strictly less than zero implying that the logical information is far from recoverable by the noisy circuit. Letting $r$ be a constant this bound shows that no noisy circuit can preserve quantum information for $d \geq d^*$ where $d^* =O(1/p)$, giving analytical justification for the above-threshold behavior observed in our numerics. Second, we show that if $r>(2e^p - 1)^{-1}$, then the coherent information is also bounded below zero. This gives an upper bound on how high the encoding rate can be before it is impossible to recover the information after applying the noisy channel.

\section{Background}
\label{sec:background}
A quantum error-correcting code maps a smaller $2^k$-dimensional logical Hilbert space into a larger $2^n$-dimensional physical Hilbert space. Any unitary on $n$ qubits can be interpreted as such a mapping by allowing $k$ of the $n$ input qubits to represent the logical state and treating the rest of the input qubits as ancillas. This map is depicted in \cref{fig:code}. For instance, if the encoding unitary $U$ is a Haar random unitary, then the codespace is a Haar random subspace. 

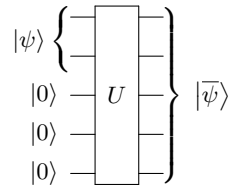
\begin{figure}
\centering
\mbox{
\Qcircuit @C=1em @R=.7em 
{\lstick{} & \multigate{4}{U} &\qw \inputgroupv{1}{2}{.8em}{1em}{\ket{\psi}}\\
\lstick{} & \ghost{U} &\qw\\
\lstick{\ket{0}} & \ghost{U} &\qw && \ket{\overline{\psi}}\\
\lstick{\ket{0}} & \ghost{U} &\qw \\
\lstick{\ket{0}} & \ghost{U} &\qw \gategroup{5}{3}{1}{3}{.7em}{\}}\\
}
}
\caption{Depiction of a general encoding circuit. In this case $\ket{\psi}$ is the $k$-qubit unencoded logical state. $\ket{\overline{\psi}}$ is the $n$-qubit encoded logical state.}
\label{fig:code}
\end{figure}

When these encoding circuits are subject to mid-circuit noise, the logical information can be partially destroyed. To quantify the error-protecting capability of these encoding circuits under mid-circuit noise, we study their coherent information. We closely follow \cite{preskill} and adopt their notation for consistency. We first define the encoding channel, which maps states in the logical Hilbert space $\mathcal H_A$ to states in the physical Hilbert space $\mathcal H_B$, as follows:
\begin{align}
    \mathcal E^{A \rightarrow B}(\rho) = \tilde{\mathcal C}(\rho \otimes\ketbra{0}^{\otimes n-k})
\end{align}
where $\tilde{\mathcal C}$ denotes the noisy circuit. To evaluate the noisy circuit's ability to transmit logical information, we would like to understand whether it is feasible to recover the original logical codestate--in particular, does there exist a decoding channel $\mathcal D^{B\rightarrow A}$ such that for any $\ket{\psi}$, $\mathcal D \circ \mathcal E(\ketbra{\psi}) = \ketbra{\psi}$. This can be quantified by first introducing a reference system $R$ and considering the maximally entangled state between $R$ and the logical qubits $A$ denoted as follows

\begin{align}
    \ket{\Phi}_{RA} = 2^{-k/2}\sum_{i=1}^{2^k} \ket{i,i},
\end{align} 
where $2^k$ is the dimension of $\mathcal H_A$. For notational purposes we also denote the density matrix of the maximally entangled state as $\Phi_{RA} := \ketbra{\Phi_{RA}}$. The quantity of interest is then the coherent information of the noisy encoding channel $\mathcal E$ acting on the logical qubits of this maximally entangled state, which we denote as $I_c(\eye/2^k,\mathcal{E})$, where $\eye/2^k$ is the reduced state of $\Phi_{RA}$ on $A$. Before formally defining the coherent information, it will be helpful to consider the Stinespring dilation of $\mathcal E^{A \rightarrow B}$ and denote it as $U^{A \rightarrow BE}$ where $E$ represents the environment. Finally, the output state of this dilation is denoted by $\ket{\phi}_{RBE}$. This purified model of our setting is depicted  in \cref{fig:stine}.
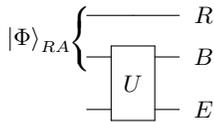
\begin{figure}
\centering
\mbox{
\Qcircuit @C=1em @R=1.25em 
{
\lstick{} & \qw & \qw & R
\inputgroupv{1}{2}{.6em}{1em}{\ket{\Phi}_{RA}\text{      }}\\
\lstick{} & \multigate{1}{U} & \qw & B\\
\lstick{} & \ghost{U} & \qw & E \\
}}
\caption{Diagram for the Stinespring dilation of the noisy circuit acting on one half of the maximally entangled state. Importantly, the coherent information of this operation characterizes how well quantum information can be recovered from the output of the noisy circuit.}
\label{fig:stine}
\end{figure}
With this setup in place, we can now define the coherent information in terms of entropies of reduced density matrices of $\ketbra{\phi}_{RBE}$: 
\begin{align}
    I_c(\eye/2^k, \mathcal{E}) = S(B) - S(E),
\end{align}
which we abbreviate by $I_c$.

First, note that this quantity is essentially equivalent to the quantum mutual information between $R$ and $B$ with respect to $\ketbra{\phi}_{RBE}$:
\begin{align}
    S(B) - S(E) &= S(B) - S(RB)\\
    &= I(R:B) - S(R) \\
    &= I(R:B) - k
\end{align}
Therefore, it is exactly the mutual information between the reference system $R$ and the physical qubits $B$, shifted by $k$.
Since $0 \leq I(R:B) \leq 2k$, we have that $-k \leq I_c \leq k$. It can be shown that if the upper bound is saturated then any logical codestate can be perfectly recovered from the noisy channel. The converse is also true in that if the purification of the maximally entangled state can be recovered, then $I_c = k$. On the other hand, it can also be shown that when $I_c = -k$ no logical information can be transmitted. Thus, the coherent information is a useful quantity for understanding how well logical information can be preserved under the noisy encoding channel. For self-containment, we present a proof sketch for these statements in \cref{app:omittedproofs}, which continues to follow John Preskill's notes \cite{preskill}. 

Another motivation for studying the coherent information is that it is closely related to the quantum channel capacity by the formula \cite{Lloyd_1997,devetak2004capacityquantumchannelsimultaneous}:
\begin{align}
\mathcal Q(\mathcal N) = \lim_{N \rightarrow \infty}\frac{1}{N} \max_\rho I_c(\rho,\mathcal N^{\otimes N}).
\end{align}

\section{2D Random Clifford circuits}
\label{sec:clifford}

We begin by studying random 2D Clifford circuits with mid-circuit depolarizing noise after every gate, where the single-qubit depolarizing channel is defined as follows:
\begin{align}
    \mathcal D_p(\rho)  = (1-p) \rho + \frac{p}{3} X \rho X^{\dagger} + \frac{p}{3} Y \rho Y^{\dagger} + \frac{p}{3} Z \rho Z^{\dagger} \label{eq:depolarizing}
\end{align}

It can be seen that these 2D Clifford circuits encode input states into stabilizer codes as follows. First,  we refer to the ancilla qubits as ``stabilizer input qubits''. Since these are initialized to $\ket{0}^{\otimes n-k}$ they are stabilized by the generators $Z_i$ where $i$ indexes each input qubit. Assuming there is no noise, the output of the circuit will be stabilized by $UZ_iU^\dagger$ where $U$ represents the Clifford circuit. Thus, the stabilizer code is defined by the generators $UZ_iU^\dagger$ for each stabilizer input qubit $i$.

\paragraph{Computational Method.}

As discussed in the previous section, our goal will be to estimate the average coherent information of noisy random Clifford circuits acting on half of the maximally entangled state. For the case of Clifford gates, we can calculate this efficiently since the maximally entangled state is a stabilizer state (it is the tensor product of $k$ Bell pairs which can each be prepared using Hadamards and CNOT gates), and so the entire procedure can be simulated using the Gottesman-Knill theorem extended to mixed states \cite{gottesman1998heisenbergrepresentationquantumcomputers,Aaronson_2004}. After applying the noisy circuit we additionally allow for a perfect syndrome measurement of the stabilizer generators, which projects the state back onto the codespace. Finally, the coherent information is calculated by computing the entropy of the output state of this protocol. An example of this entire procedure is depicted in \cref{fig:CI}.

\begin{figure}
\centering
\mbox{
\Qcircuit @C=1em @R=.7em 
{&& \lstick{\text{R     } \ket{0}} & \gate{H} & \ctrl{1} & \qw & \qw & \qw  & \qw\\
&&\lstick{\ket{0}} & \qw  &\targ & \qw & \multigate{3}{\tilde{\mathcal C}} & \multigate{3}{\mathcal S} & \qw \\
&&\lstick{\ket{0}} & \qw & \qw & \qw & \ghost{\mathcal E} & \ghost{\mathcal S}  & \qw &&& \lstick{I_c}\\
&&\lstick{\ket{0}} & \qw & \qw & \qw & \ghost{\mathcal E} & \ghost{\mathcal S}  & \qw \\
&&\lstick{\ket{0}} & \qw & \qw & \qw & \ghost{\mathcal E} & \ghost{\mathcal S}  & \qw  \gategroup{5}{9}{1}{9}{.7em}{\}}
}}
\caption{Diagram of our numerical simulation. Qubit 1 represents the reference qubit and is prepared in a Bell pair with qubit 2, which represents the logical qubit (here, there is only one logical qubit and so $k=1$). Qubits 3-5 represent the ancilla, or stabilizer, qubits. The noisy Clifford circuit is applied to the data qubits (qubits 2-5) which is denoted by $\tilde{\mathcal C}$. Next, the stabilizer generators of the code are measured perfectly which is denoted by $\mathcal S$. Finally, the coherent information of this channel is calculated.}
\label{fig:CI}
\end{figure}
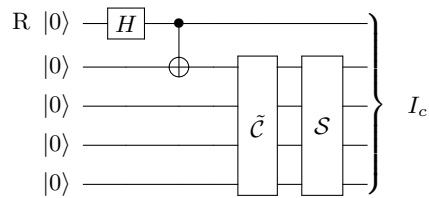

\paragraph{Numerical Setup.}

For each depth between depth $10$ and $24$, the simulation is repeated while sweeping over a wide range of error rates. In addition, the system size is varied as $n=16,24,144,256,$ and $400$. Furthermore, a code rate of $1/8$ is chosen, and logical qubits are placed evenly throughout the 2D grid. For each choice of depth, error rate, and system size, we collect $1000$ samples and compute the average coherent information. 

\paragraph{Results.}

It is evident from \cref{fig:clifford}a that there is a clear crossing point where the coherent information for each system size all converge to zero at a critical error rate. In particular, it can be seen that before this critical error rate, the coherent information is positive and continues to increase with larger system size. Therefore, this parameter regime is in a phase in which error correction is possible. However, after the critical error rate, the coherent information is negative and increasing the system size leads to even lower coherent information. This indicates that the circuit is now in a phase in which error correction is no longer possible. To further investigate this phase transition, a scaling collapse is plotted which uses the following dimensionless scaling ansatz, whose form is motivated by the statistical mechanics of phase transitions \cite{RevModPhys.70.653}:
\begin{align}
    I_c = f(n^\lambda(p-p^*))
\end{align}
We can estimate the critical exponent $\lambda$ and the critical error rate $p^*$ by approximating $f(x)$ as $f(x) = A+Bx+Cx^2$ using its Taylor expansion around $x=0$. Fitting these five parameters, $A$, $B$, $C$, $p^*$, and $\lambda$, to the data results in an estimate for $p^*$ and $\lambda$, which can then be used to plot $I_c$ with respect to our dimensionless ansatz. From \cref{fig:clifford}b it is apparent that our scaling ansatz faithfully parameterizes the behavior of $I_c$ since all of the data points fall on a single straight line regardless of the system size.
\begin{figure}
    \centering
    \includegraphics[width=\linewidth]{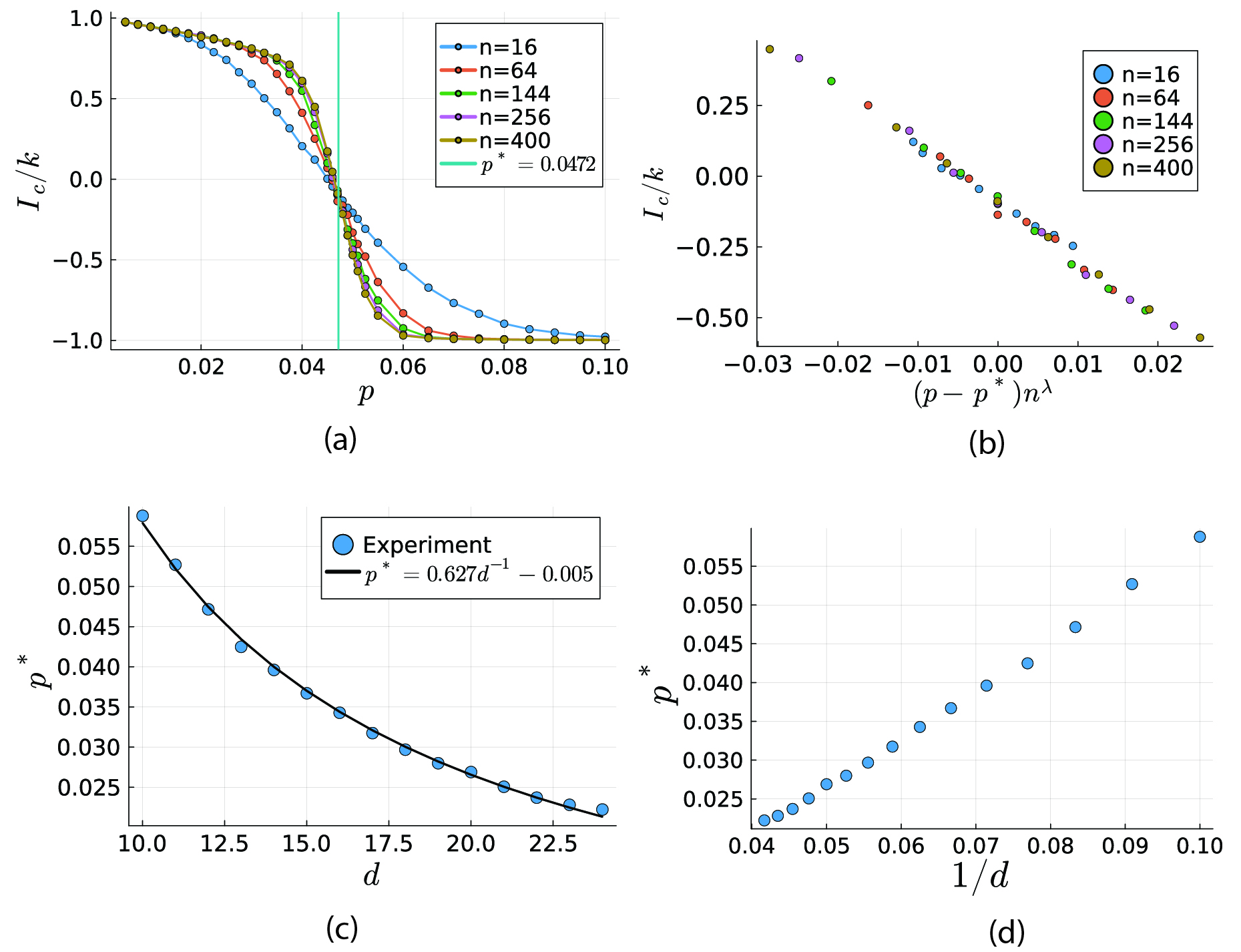}
\caption{Numerical results for 2D random brickwork Clifford circuits with periodic boundary conditions subjected to mid-circuit depolarizing noise. In all cases an encoding rate of $r = 1/8$ is used. (a) Coherent information per logical qubit is displayed with respect to a range of error rates $p$ between $0.005$ and $0.1$. System size is scaled from $n=16$ and $n=400$. The circuit depth is fixed at $d=12$. The critical error rate is depicted by the vertical line which occurs at $p^* = 0.0472$. (b) A scaling collapse for the data in (a) is plotted for error rates near the critical point. The scaling ansatz uses the estimated critical error rate $p^* = 0.0472$ and critical exponent $\lambda = 0.616$. (c) The plots in (a) and (b) are repeated for circuit depths $d$ ranging from $10$ to $24$. Linear regression on the inverse depth is used to estimate the relationship that $p^* = 0.627d^{-1} -0.005$. (d) Finally, the critical error rate $p^*$ is plotted with respect to the inverse depth.}
\label{fig:clifford}
\end{figure}

\section{All-to-All Haar Random circuits}
\label{sec:haar}
Next, we conduct similar numerics for the case of Haar random circuits with all-to-all connectivity. This is not feasible by direct simulation but can be done using a standard second moment approximation \cite{PhysRevX.7.031016,PhysRevX.8.021014,Zhou_2019}. 

\paragraph{Computational Method.} In the definition of coherent information, we first replace the von Neumann entropy with the Renyi-2 entropy: $S_2(A) = -\log \Tr \rho_A^2$. This gives:

\begin{align}
    I_{c,2} &= S_2(B) - S_2(RB) \\
    &= -\log(\Tr \rho_B^2) + \log(\Tr \rho^2),
\end{align}
where $\rho = (\mathrm{Id}_R \otimes \mathcal E)(\ketbra{\Phi_{RA}})$ represents the marginal state on $RB$ of $\ket{\phi_{RBE}}$.
Next, we average over the random circuit instances and move the expectation over circuit instances $\mc C$ inside of the log.
\begin{align}
    \E_{\mathcal C} I_{c,2} = -\E_{\mathcal C}\log(\Tr \rho_B^2) + \E_{\mathcal C} \log(\Tr \rho^2) \\
    \label{eq:Ip}
    \rightarrow I_p = -\log( \E_{\mathcal C} \Tr \rho_B^2) + \log( \E_{\mathcal C} \Tr\rho^2),
\end{align}
Notice that this expression is now a function of second moment operators of the Haar measure, which are tractable to compute in the special case where the gates have all-to-all connectivity. In addition, it is possible to perform these calculations for any single-qubit noise channels and so we present results for both depolarizing noise and amplitude damping noise. The single-qubit depolarizing channel is defined in \cref{eq:depolarizing}. The amplitude damping channel is defined as follows:
\begin{align}
    \mathcal A_p(\rho)= A_0\rho A_0^\dagger + A_1 \rho A_1^\dagger \label{eq:damping}
\end{align}
where 
\begin{align}
    A_0 = \begin{bmatrix}
        1 & 0 \\
        0 & \sqrt{1-p}
    \end{bmatrix}
\end{align}
and 
\begin{align}
    A_1 = \begin{bmatrix}
        0 & \sqrt{p} \\
        0 & 0
    \end{bmatrix}
\end{align}
To follow the full calculation for both choices of noise see \cref{app:purity}.

\paragraph{Numerical Setup.}

Once again, we perform the numerics over a wide range of depths and noise strengths. As in the Clifford case, we provide a plot of the scaling collapse with respect to a dimensionless scaling ansatz. In our numerics, we again consider a rate of $r = 1/8$ and consider system sizes of $n=10$ to $n=50$. 

\paragraph{Results.}
We present the results for the case of depolarizing noise in \cref{fig:haardepol} and the results for amplitude damping noise in \cref{fig:haaramp}. When depth is fixed it can be seen that the coherent information once again undergoes a phase transition as the error rate is scaled. Note that the same would be true if we instead fixed the error rate and scaled the depth. In the case of depolarizing noise, the critical point occurs at approximately $p^*d^* \sim 0.39$ while in the amplitude damping case, the critical point occurs at $p^*d^* \sim 0.74$. This suggests that the circuits are more robust to amplitude damping noise than depolarizing. In fact, for a given error rate, information can be preserved for almost twice as long for amplitude damping noise as opposed to depolarizing noise. Interestingly, both noise models have similar critical exponents: $\lambda = 0.467$ in the depolarizing case and $\lambda = 0.525$ in the amplitude damping case. Comparing the 2D numerics in the prior section to the all-to-all case, it is worthwhile to note that for depolarizing noise the 2D circuits have a higher critical point of $p^*d^* \sim 0.627$ compared to $p^*d^* \sim 0.39$ in the all-to-all case. Thus, the extra degrees of interaction in the all-to-all case seems to have spread the noise more quickly rather than encode the information faster. The 2D case also results in a higher critical exponent of $\lambda = 0.616$.
\begin{figure}
    \centering
    \includegraphics[width=\linewidth]{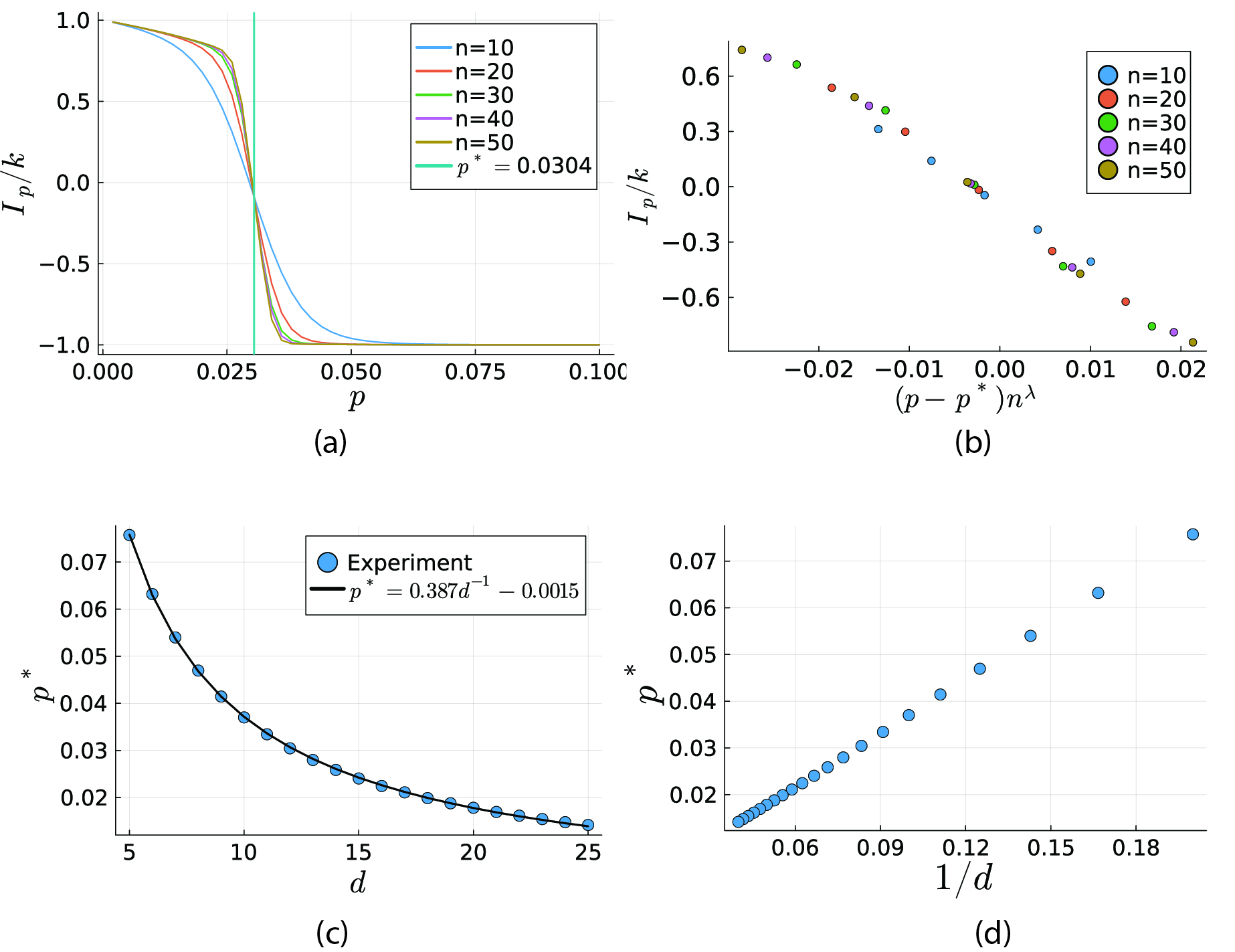}
\caption{Numerical results for all-to-all Haar random circuits subjected to mid-circuit depolarizing noise. In all cases an encoding rate of $r = 1/8$ is used. (a) Log-average purity per logical qubit is displayed with respect to a range of error rates $p$ between $0.002$ and $0.1$. System size is scaled from $n=10$ and $n=50$. The circuit depth is fixed at $d=12$. The critical error rate is depicted by the vertical line which occurs at $p^* = 0.0304$. (b) A scaling collapse for the data in (a) is plotted for error rates near the critical point. The scaling ansatz uses the estimated critical error rate $p^* = 0.0304$ and critical exponent $\lambda = 0.467$. (c) The plots in (a) and (b) are repeated for circuit depths $d$ ranging from $5$ to $25$. Linear regression on the inverse depth is used to estimate the relationship that $p^* = 0.387d^{-1} -0.0015$. (d) Finally, the critical error rate $p^*$ is plotted with respect to the inverse depth.}
\label{fig:haardepol}
\end{figure}

\begin{figure}
    \centering
    \includegraphics[width=\linewidth]{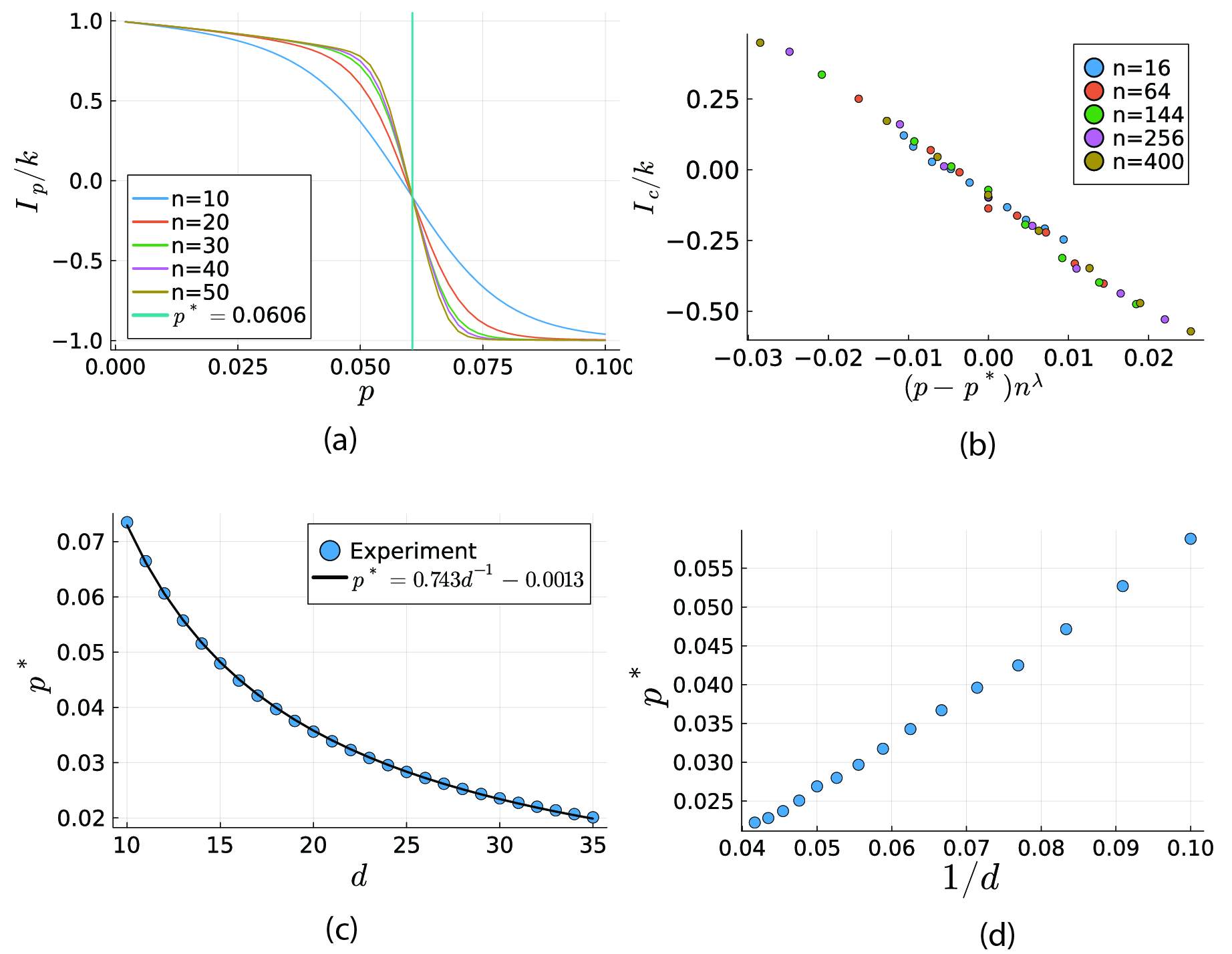}
\caption{Numerical results for all-to-all Haar random circuits subjected to mid-circuit amplitude damping noise. In all cases an encoding rate of $r = 1/8$ is used. (a) Log-average purity per logical qubit is displayed with respect to a range of error rates $p$ between $0.002$ and $0.1$. System size is scaled from $n=10$ and $n=50$. The circuit depth is fixed at $d=12$. The critical error rate is depicted by the vertical line which occurs at $p^* = 0.0606$. (b) A scaling collapse for the data in (a) is plotted for error rates near the critical point. The scaling ansatz uses the estimated critical error rate $p^* = 0.0606$ and critical exponent $\lambda = 0.525$. (c) The plots in (a) and (b) are repeated for circuit depths $d$ ranging from $10$ to $35$. Linear regression on the inverse depth is used to estimate the relationship that $p^* = 0.743d^{-1} -0.0013$. (d) Finally, the critical error rate $p^*$ is plotted with respect to the inverse depth.}
\label{fig:haaramp}
\end{figure}

\section{Worst-case circuits}
\label{sec:worstcase}
In order to support these numerical results, we analytically prove a rigorous upper bound on the coherent information of \textit{any} noisy circuit, which asymptotically matches our numerics for noisy random circuits. Note these bounds apply to circuits on any architecture and even allow for arbitrary non-local gates. This suggests that noisy random circuits achieve the optimal asymptotic trade-off between error rate and depth.
\begin{theorem} \label{theorem:worst-case}
    Let $\tilde{\mc C}$ be an arbitrary quantum circuit acting on $n$ qubits, which has $d$ layers of interspersed depolarizing noise of strength $p$ acting independently on each qubit. Then the coherent information of the encoding channel which encodes a $k$-qubit quantum state $\zeta$ as $\mc E(\zeta) = \tilde{C}(\zeta \otimes \ketbra{0}^{n-k})$, is bounded as follows:
    \begin{align}
        I_c(\frac{\eye}{2^k}, \mc E) \leq e^{-p d}(n+k)-k
    \end{align}
\end{theorem}
We obtain this bound using relative entropy convergence under tensor products of depolarizing channels \cite{JonJoel,M_ller_Hermes_2016}. The proof is deferred to \cref{app:worstcaseproof}.

Let us denote the encoding rate as $r:=k/n$. This implies that when $d > p^{-1}\log(1+1/r)$, then $I_c < 0$. Assuming $r$ is some constant, this proves that the coherent information is negative and thus the logical information is unrecoverable for any noisy circuit after a depth threshold of $d^* = O(p^{-1})$. Furthermore, setting $d=1$ shows that $r < (e^p - 1)^{-1}$ in order for $I_c \geq 0$. This provides a bound on the largest allowable encoding rate that is still recoverable after any depth. 

It is worthwhile to sanity check this bound against known fault-tolerance results. In particular, Aharonov \textit{et al.}~\cite{aharonov1996limitationsnoisyreversiblecomputation} give a protocol (that is closely related to \cite{aharonov1999faulttolerantquantumcomputationconstant}) that is robust to errors and uses at least $k2^d$ ancilla qubits. This corresponds to a rate of $r \leq 2^{-d}$. Plugging this into our bound shows that the coherent information of their circuit becomes negative after a depth of $p^{-1}\log(1+1/r) \geq p^{-1}d>d$. Thus, their protocol uses just enough ancilla qubits to avoid our bounds.
\\
\section{Discussion}
Our results suggest several avenues for future work. For instance, one immediate goal is to prove the existence of our observed phase transition analytically. Similar results have been achieved for low-depth random-circuit codes in the code capacity setting \cite{liu2025approximatequantumerrorcorrection} and so it is possible that related techniques could be extended to our setting of random circuits subject to mid-circuit depolarizing noise.

Next, our work has tightly characterized the tradeoff between depth and error rate, but it remains to understand how encoding rate fits into this picture. For instance, up until what depth does the circuit preserve a vanishing encoding rate? It would also be interesting to investigate how our results change as the circuit connectivity is varied between 2D and all-to-all architectures (e.g. varying the degree of the connectivity graph). We hope to further expand our phase diagram to include these factors.

Finally, a constant-depth phase transition in quantum circuits with depolarizing noise has been previously observed in a different setting of computational complexity. In particular, the output distributions of noisy IQP and Clifford-magic circuits \cite{Rajakumar_2025,nelson2024polynomialtimeclassicalsimulationnoisy} can be efficiently classically sampled above a critical depth of approximately $d^* = \tilde{O}(1/p)$ (omitting log factors), whereas below this critical depth, the task of sampling is outside the polynomial hierarchy \cite{fujii_computational_2016}, indicating there is no efficient classical algorithm. These prior works exploit a noise percolation phenomenon beyond the critical depth, which cuts the output state into many small, unentangled pieces. Our work shows a similar phenomenon: namely, we show an inability to effectively encode and protect quantum information past the critical depth threshold, which also implies a breakdown of long-range entanglement. It would be interesting to see whether this connection could be made more rigorous. 
\\
\section*{Acknowledgements}
This material is based upon work supported by the U.S. Department of Energy, Office of Science, Accelerated Research in Quantum Computing, Fundamental Algorithmic Research toward Quantum Utility (FAR-Qu). We thank Dominik Hangleiter, Zhi-Yuan Wei, Daniel Malz and Alexey Gorshkov for helpful discussions.
This material is based upon work supported in part by the NSF QLCI award OMA2120757. This work was performed in
part at the Kavli Institute for Theoretical Physics (KITP), which is supported by grant NSF PHY-2309135. JN is supported by the National Science Foundation Graduate Research Fellowship Program under Grant No. DGE 2236417.

\bibliography{Bibliography.bib, doms_refs.bib}

\begin{thebibliography}{34}%
\makeatletter
\providecommand \@ifxundefined [1]{%
 \@ifx{#1\undefined}
}%
\providecommand \@ifnum [1]{%
 \ifnum #1\expandafter \@firstoftwo
 \else \expandafter \@secondoftwo
 \fi
}%
\providecommand \@ifx [1]{%
 \ifx #1\expandafter \@firstoftwo
 \else \expandafter \@secondoftwo
 \fi
}%
\providecommand \natexlab [1]{#1}%
\providecommand \enquote  [1]{``#1''}%
\providecommand \bibnamefont  [1]{#1}%
\providecommand \bibfnamefont [1]{#1}%
\providecommand \citenamefont [1]{#1}%
\providecommand \href@noop [0]{\@secondoftwo}%
\providecommand \href [0]{\begingroup \@sanitize@url \@href}%
\providecommand \@href[1]{\@@startlink{#1}\@@href}%
\providecommand \@@href[1]{\endgroup#1\@@endlink}%
\providecommand \@sanitize@url [0]{\catcode `\\12\catcode `\$12\catcode `\&12\catcode `\#12\catcode `\^12\catcode `\_12\catcode `\%12\relax}%
\providecommand \@@startlink[1]{}%
\providecommand \@@endlink[0]{}%
\providecommand \url  [0]{\begingroup\@sanitize@url \@url }%
\providecommand \@url [1]{\endgroup\@href {#1}{\urlprefix }}%
\providecommand \urlprefix  [0]{URL }%
\providecommand \Eprint [0]{\href }%
\providecommand \doibase [0]{http://dx.doi.org/}%
\providecommand \selectlanguage [0]{\@gobble}%
\providecommand \bibinfo  [0]{\@secondoftwo}%
\providecommand \bibfield  [0]{\@secondoftwo}%
\providecommand \translation [1]{[#1]}%
\providecommand \BibitemOpen [0]{}%
\providecommand \bibitemStop [0]{}%
\providecommand \bibitemNoStop [0]{.\EOS\space}%
\providecommand \EOS [0]{\spacefactor3000\relax}%
\providecommand \BibitemShut  [1]{\csname bibitem#1\endcsname}%
\let\auto@bib@innerbib\@empty
\bibitem [{\citenamefont {Li}\ \emph {et~al.}(2019)\citenamefont {Li}, \citenamefont {Chen},\ and\ \citenamefont {Fisher}}]{PhysRevB.100.134306}%
  \BibitemOpen
  \bibfield  {author} {\bibinfo {author} {\bibfnamefont {Y.}~\bibnamefont {Li}}, \bibinfo {author} {\bibfnamefont {X.}~\bibnamefont {Chen}}, \ and\ \bibinfo {author} {\bibfnamefont {M.~P.~A.}\ \bibnamefont {Fisher}},\ }\href {\doibase 10.1103/PhysRevB.100.134306} {\bibfield  {journal} {\bibinfo  {journal} {Phys. Rev. B}\ }\textbf {\bibinfo {volume} {100}},\ \bibinfo {pages} {134306} (\bibinfo {year} {2019})}\BibitemShut {NoStop}%
\bibitem [{\citenamefont {Skinner}\ \emph {et~al.}(2019)\citenamefont {Skinner}, \citenamefont {Ruhman},\ and\ \citenamefont {Nahum}}]{PhysRevX.9.031009}%
  \BibitemOpen
  \bibfield  {author} {\bibinfo {author} {\bibfnamefont {B.}~\bibnamefont {Skinner}}, \bibinfo {author} {\bibfnamefont {J.}~\bibnamefont {Ruhman}}, \ and\ \bibinfo {author} {\bibfnamefont {A.}~\bibnamefont {Nahum}},\ }\href {\doibase 10.1103/PhysRevX.9.031009} {\bibfield  {journal} {\bibinfo  {journal} {Phys. Rev. X}\ }\textbf {\bibinfo {volume} {9}},\ \bibinfo {pages} {031009} (\bibinfo {year} {2019})}\BibitemShut {NoStop}%
\bibitem [{\citenamefont {Li}\ \emph {et~al.}(2018)\citenamefont {Li}, \citenamefont {Chen},\ and\ \citenamefont {Fisher}}]{PhysRevB.98.205136}%
  \BibitemOpen
  \bibfield  {author} {\bibinfo {author} {\bibfnamefont {Y.}~\bibnamefont {Li}}, \bibinfo {author} {\bibfnamefont {X.}~\bibnamefont {Chen}}, \ and\ \bibinfo {author} {\bibfnamefont {M.~P.~A.}\ \bibnamefont {Fisher}},\ }\href {\doibase 10.1103/PhysRevB.98.205136} {\bibfield  {journal} {\bibinfo  {journal} {Phys. Rev. B}\ }\textbf {\bibinfo {volume} {98}},\ \bibinfo {pages} {205136} (\bibinfo {year} {2018})}\BibitemShut {NoStop}%
\bibitem [{\citenamefont {Gullans}\ and\ \citenamefont {Huse}(2020)}]{Gullans_2020}%
  \BibitemOpen
  \bibfield  {author} {\bibinfo {author} {\bibfnamefont {M.~J.}\ \bibnamefont {Gullans}}\ and\ \bibinfo {author} {\bibfnamefont {D.~A.}\ \bibnamefont {Huse}},\ }\href {\doibase 10.1103/physrevx.10.041020} {\bibfield  {journal} {\bibinfo  {journal} {Physical Review X}\ }\textbf {\bibinfo {volume} {10}} (\bibinfo {year} {2020}),\ 10.1103/physrevx.10.041020}\BibitemShut {NoStop}%
\bibitem [{\citenamefont {Weinstein}\ \emph {et~al.}(2023)\citenamefont {Weinstein}, \citenamefont {Kelly}, \citenamefont {Marino},\ and\ \citenamefont {Altman}}]{Weinstein_2023}%
  \BibitemOpen
  \bibfield  {author} {\bibinfo {author} {\bibfnamefont {Z.}~\bibnamefont {Weinstein}}, \bibinfo {author} {\bibfnamefont {S.~P.}\ \bibnamefont {Kelly}}, \bibinfo {author} {\bibfnamefont {J.}~\bibnamefont {Marino}}, \ and\ \bibinfo {author} {\bibfnamefont {E.}~\bibnamefont {Altman}},\ }\href {\doibase 10.1103/physrevlett.131.220404} {\bibfield  {journal} {\bibinfo  {journal} {Physical Review Letters}\ }\textbf {\bibinfo {volume} {131}} (\bibinfo {year} {2023}),\ 10.1103/physrevlett.131.220404}\BibitemShut {NoStop}%
\bibitem [{\citenamefont {Qian}\ and\ \citenamefont {Wang}(2024)}]{qian2024coherentinformationphasetransition}%
  \BibitemOpen
  \bibfield  {author} {\bibinfo {author} {\bibfnamefont {D.}~\bibnamefont {Qian}}\ and\ \bibinfo {author} {\bibfnamefont {J.}~\bibnamefont {Wang}},\ }\href {https://arxiv.org/abs/2408.16267} {\enquote {\bibinfo {title} {Coherent information phase transition in a noisy quantum circuit},}\ } (\bibinfo {year} {2024}),\ \Eprint {http://arxiv.org/abs/2408.16267} {arXiv:2408.16267 [quant-ph]} \BibitemShut {NoStop}%
\bibitem [{\citenamefont {Nelson}\ \emph {et~al.}(2025{\natexlab{a}})\citenamefont {Nelson}, \citenamefont {Bentsen}, \citenamefont {Flammia},\ and\ \citenamefont {Gullans}}]{nelsonrandomcircuitcodes}%
  \BibitemOpen
  \bibfield  {author} {\bibinfo {author} {\bibfnamefont {J.}~\bibnamefont {Nelson}}, \bibinfo {author} {\bibfnamefont {G.}~\bibnamefont {Bentsen}}, \bibinfo {author} {\bibfnamefont {S.~T.}\ \bibnamefont {Flammia}}, \ and\ \bibinfo {author} {\bibfnamefont {M.~J.}\ \bibnamefont {Gullans}},\ }\href {\doibase 10.1103/PhysRevResearch.7.013040} {\bibfield  {journal} {\bibinfo  {journal} {Phys. Rev. Res.}\ }\textbf {\bibinfo {volume} {7}},\ \bibinfo {pages} {013040} (\bibinfo {year} {2025}{\natexlab{a}})}\BibitemShut {NoStop}%
\bibitem [{Note1()}]{Note1}%
  \BibitemOpen
  \bibinfo {note} {In this paper, we adopt the standard asymptotic notation $O$, $\Omega $, $\Theta $, $o$, and $\omega $ to describe relationships between the growth rates of two nonnegative functions $f(n)$ and $g(n)$ as $n \to \infty $. We write $f = O(g)$ when $\lim _{n \to \infty } f(n)/g(n) < \infty $, which is equivalent to stating that $g = \Omega (f)$. If $\lim _{n \to \infty } f(n)/g(n) = 0$, we denote this by $f = o(g)$, or equivalently $g = \omega (f)$. The notation $f = \Theta (g)$ indicates that both $f = O(g)$ and $f = \Omega (g)$ are satisfied. When a tilde is used, as in $\protect \tilde {\Theta }$, it signifies that multiplicative polylogarithmic factors in $n$ are omitted from the scaling.}\BibitemShut {Stop}%
\bibitem [{\citenamefont {Müller-Hermes}\ \emph {et~al.}(2016)\citenamefont {Müller-Hermes}, \citenamefont {Stilck~França},\ and\ \citenamefont {Wolf}}]{M_ller_Hermes_2016}%
  \BibitemOpen
  \bibfield  {author} {\bibinfo {author} {\bibfnamefont {A.}~\bibnamefont {Müller-Hermes}}, \bibinfo {author} {\bibfnamefont {D.}~\bibnamefont {Stilck~França}}, \ and\ \bibinfo {author} {\bibfnamefont {M.~M.}\ \bibnamefont {Wolf}},\ }\href {\doibase 10.1063/1.4939560} {\bibfield  {journal} {\bibinfo  {journal} {Journal of Mathematical Physics}\ }\textbf {\bibinfo {volume} {57}} (\bibinfo {year} {2016}),\ 10.1063/1.4939560}\BibitemShut {NoStop}%
\bibitem [{\citenamefont {Aharonov}\ \emph {et~al.}(1996)\citenamefont {Aharonov}, \citenamefont {Ben-Or}, \citenamefont {Impagliazzo},\ and\ \citenamefont {Nisan}}]{aharonov1996limitationsnoisyreversiblecomputation}%
  \BibitemOpen
  \bibfield  {author} {\bibinfo {author} {\bibfnamefont {D.}~\bibnamefont {Aharonov}}, \bibinfo {author} {\bibfnamefont {M.}~\bibnamefont {Ben-Or}}, \bibinfo {author} {\bibfnamefont {R.}~\bibnamefont {Impagliazzo}}, \ and\ \bibinfo {author} {\bibfnamefont {N.}~\bibnamefont {Nisan}},\ }\href {https://arxiv.org/abs/quant-ph/9611028} {\enquote {\bibinfo {title} {Limitations of noisy reversible computation},}\ } (\bibinfo {year} {1996}),\ \Eprint {http://arxiv.org/abs/quant-ph/9611028} {arXiv:quant-ph/9611028 [quant-ph]} \BibitemShut {NoStop}%
\bibitem [{Note2()}]{Note2}%
  \BibitemOpen
  \bibinfo {note} {Optimal, here, is with respect to the asymptotic relationship between depth and noise strength and does not take into account other parameters such as encoding rate.}\BibitemShut {Stop}%
\bibitem [{\citenamefont {Brown}\ and\ \citenamefont {Fawzi}(2013{\natexlab{a}})}]{brown2013short}%
  \BibitemOpen
  \bibfield  {author} {\bibinfo {author} {\bibfnamefont {W.}~\bibnamefont {Brown}}\ and\ \bibinfo {author} {\bibfnamefont {O.}~\bibnamefont {Fawzi}},\ }in\ \href {\doibase 10.1109/ISIT.2013.6620245} {\emph {\bibinfo {booktitle} {2013 IEEE International Symposium on Information Theory}}}\ (\bibinfo {year} {2013})\ pp.\ \bibinfo {pages} {346--350}\BibitemShut {NoStop}%
\bibitem [{\citenamefont {Brown}\ and\ \citenamefont {Fawzi}(2015)}]{brown2015decoupling}%
  \BibitemOpen
  \bibfield  {author} {\bibinfo {author} {\bibfnamefont {W.}~\bibnamefont {Brown}}\ and\ \bibinfo {author} {\bibfnamefont {O.}~\bibnamefont {Fawzi}},\ }\href@noop {} {\bibfield  {journal} {\bibinfo  {journal} {Communications in mathematical physics}\ }\textbf {\bibinfo {volume} {340}},\ \bibinfo {pages} {867} (\bibinfo {year} {2015})}\BibitemShut {NoStop}%
\bibitem [{\citenamefont {Gullans}\ \emph {et~al.}(2021)\citenamefont {Gullans}, \citenamefont {Krastanov}, \citenamefont {Huse}, \citenamefont {Jiang},\ and\ \citenamefont {Flammia}}]{gullans2021quantum}%
  \BibitemOpen
  \bibfield  {author} {\bibinfo {author} {\bibfnamefont {M.~J.}\ \bibnamefont {Gullans}}, \bibinfo {author} {\bibfnamefont {S.}~\bibnamefont {Krastanov}}, \bibinfo {author} {\bibfnamefont {D.~A.}\ \bibnamefont {Huse}}, \bibinfo {author} {\bibfnamefont {L.}~\bibnamefont {Jiang}}, \ and\ \bibinfo {author} {\bibfnamefont {S.~T.}\ \bibnamefont {Flammia}},\ }\href {\doibase 10.1103/PhysRevX.11.031066} {\bibfield  {journal} {\bibinfo  {journal} {Phys. Rev. X}\ }\textbf {\bibinfo {volume} {11}},\ \bibinfo {pages} {031066} (\bibinfo {year} {2021})}\BibitemShut {NoStop}%
\bibitem [{\citenamefont {Hayden}\ \emph {et~al.}(2007)\citenamefont {Hayden}, \citenamefont {Horodecki}, \citenamefont {Winter},\ and\ \citenamefont {Yard}}]{hayden2008decoupling}%
  \BibitemOpen
  \bibfield  {author} {\bibinfo {author} {\bibfnamefont {P.~M.}\ \bibnamefont {Hayden}}, \bibinfo {author} {\bibfnamefont {M.}~\bibnamefont {Horodecki}}, \bibinfo {author} {\bibfnamefont {A.~J.}\ \bibnamefont {Winter}}, \ and\ \bibinfo {author} {\bibfnamefont {J.~T.}\ \bibnamefont {Yard}},\ }\href {https://api.semanticscholar.org/CorpusID:15027859} {\bibfield  {journal} {\bibinfo  {journal} {Open Syst. Inf. Dyn.}\ }\textbf {\bibinfo {volume} {15}},\ \bibinfo {pages} {7} (\bibinfo {year} {2007})}\BibitemShut {NoStop}%
\bibitem [{\citenamefont {Brown}\ and\ \citenamefont {Fawzi}(2013{\natexlab{b}})}]{brown2013scramblingspeedrandomquantum}%
  \BibitemOpen
  \bibfield  {author} {\bibinfo {author} {\bibfnamefont {W.}~\bibnamefont {Brown}}\ and\ \bibinfo {author} {\bibfnamefont {O.}~\bibnamefont {Fawzi}},\ }\href {https://arxiv.org/abs/1210.6644} {\enquote {\bibinfo {title} {Scrambling speed of random quantum circuits},}\ } (\bibinfo {year} {2013}{\natexlab{b}}),\ \Eprint {http://arxiv.org/abs/1210.6644} {arXiv:1210.6644 [quant-ph]} \BibitemShut {NoStop}%
\bibitem [{\citenamefont {Lloyd}(1997)}]{Lloyd_1997}%
  \BibitemOpen
  \bibfield  {author} {\bibinfo {author} {\bibfnamefont {S.}~\bibnamefont {Lloyd}},\ }\href {\doibase 10.1103/physreva.55.1613} {\bibfield  {journal} {\bibinfo  {journal} {Physical Review A}\ }\textbf {\bibinfo {volume} {55}},\ \bibinfo {pages} {1613–1622} (\bibinfo {year} {1997})}\BibitemShut {NoStop}%
\bibitem [{\citenamefont {Devetak}\ and\ \citenamefont {Shor}(2004)}]{devetak2004capacityquantumchannelsimultaneous}%
  \BibitemOpen
  \bibfield  {author} {\bibinfo {author} {\bibfnamefont {I.}~\bibnamefont {Devetak}}\ and\ \bibinfo {author} {\bibfnamefont {P.~W.}\ \bibnamefont {Shor}},\ }\href {https://arxiv.org/abs/quant-ph/0311131} {\enquote {\bibinfo {title} {The capacity of a quantum channel for simultaneous transmission of classical and quantum information},}\ } (\bibinfo {year} {2004}),\ \Eprint {http://arxiv.org/abs/quant-ph/0311131} {arXiv:quant-ph/0311131 [quant-ph]} \BibitemShut {NoStop}%
\bibitem [{\citenamefont {Nahum}\ \emph {et~al.}(2017)\citenamefont {Nahum}, \citenamefont {Ruhman}, \citenamefont {Vijay},\ and\ \citenamefont {Haah}}]{PhysRevX.7.031016}%
  \BibitemOpen
  \bibfield  {author} {\bibinfo {author} {\bibfnamefont {A.}~\bibnamefont {Nahum}}, \bibinfo {author} {\bibfnamefont {J.}~\bibnamefont {Ruhman}}, \bibinfo {author} {\bibfnamefont {S.}~\bibnamefont {Vijay}}, \ and\ \bibinfo {author} {\bibfnamefont {J.}~\bibnamefont {Haah}},\ }\href {\doibase 10.1103/PhysRevX.7.031016} {\bibfield  {journal} {\bibinfo  {journal} {Phys. Rev. X}\ }\textbf {\bibinfo {volume} {7}},\ \bibinfo {pages} {031016} (\bibinfo {year} {2017})}\BibitemShut {NoStop}%
\bibitem [{\citenamefont {Nahum}\ \emph {et~al.}(2018)\citenamefont {Nahum}, \citenamefont {Vijay},\ and\ \citenamefont {Haah}}]{PhysRevX.8.021014}%
  \BibitemOpen
  \bibfield  {author} {\bibinfo {author} {\bibfnamefont {A.}~\bibnamefont {Nahum}}, \bibinfo {author} {\bibfnamefont {S.}~\bibnamefont {Vijay}}, \ and\ \bibinfo {author} {\bibfnamefont {J.}~\bibnamefont {Haah}},\ }\href {\doibase 10.1103/PhysRevX.8.021014} {\bibfield  {journal} {\bibinfo  {journal} {Phys. Rev. X}\ }\textbf {\bibinfo {volume} {8}},\ \bibinfo {pages} {021014} (\bibinfo {year} {2018})}\BibitemShut {NoStop}%
\bibitem [{\citenamefont {Zhou}\ and\ \citenamefont {Nahum}(2019)}]{Zhou_2019}%
  \BibitemOpen
  \bibfield  {author} {\bibinfo {author} {\bibfnamefont {T.}~\bibnamefont {Zhou}}\ and\ \bibinfo {author} {\bibfnamefont {A.}~\bibnamefont {Nahum}},\ }\href {\doibase 10.1103/physrevb.99.174205} {\bibfield  {journal} {\bibinfo  {journal} {Physical Review B}\ }\textbf {\bibinfo {volume} {99}} (\bibinfo {year} {2019}),\ 10.1103/physrevb.99.174205}\BibitemShut {NoStop}%
\bibitem [{\citenamefont {Preskill}()}]{preskill}%
  \BibitemOpen
  \bibfield  {author} {\bibinfo {author} {\bibfnamefont {J.}~\bibnamefont {Preskill}},\ }\href {https://www.preskill.caltech.edu/ph219/} {\enquote {\bibinfo {title} {Physics219 caltech lecture notes},}\ }\BibitemShut {NoStop}%
\bibitem [{\citenamefont {Gottesman}(1998)}]{gottesman1998heisenbergrepresentationquantumcomputers}%
  \BibitemOpen
  \bibfield  {author} {\bibinfo {author} {\bibfnamefont {D.}~\bibnamefont {Gottesman}},\ }\href {https://arxiv.org/abs/quant-ph/9807006} {\enquote {\bibinfo {title} {The heisenberg representation of quantum computers},}\ } (\bibinfo {year} {1998}),\ \Eprint {http://arxiv.org/abs/quant-ph/9807006} {arXiv:quant-ph/9807006 [quant-ph]} \BibitemShut {NoStop}%
\bibitem [{\citenamefont {Aaronson}\ and\ \citenamefont {Gottesman}(2004)}]{Aaronson_2004}%
  \BibitemOpen
  \bibfield  {author} {\bibinfo {author} {\bibfnamefont {S.}~\bibnamefont {Aaronson}}\ and\ \bibinfo {author} {\bibfnamefont {D.}~\bibnamefont {Gottesman}},\ }\href {\doibase 10.1103/physreva.70.052328} {\bibfield  {journal} {\bibinfo  {journal} {Physical Review A}\ }\textbf {\bibinfo {volume} {70}} (\bibinfo {year} {2004}),\ 10.1103/physreva.70.052328}\BibitemShut {NoStop}%
\bibitem [{\citenamefont {Fisher}(1998)}]{RevModPhys.70.653}%
  \BibitemOpen
  \bibfield  {author} {\bibinfo {author} {\bibfnamefont {M.~E.}\ \bibnamefont {Fisher}},\ }\href {\doibase 10.1103/RevModPhys.70.653} {\bibfield  {journal} {\bibinfo  {journal} {Rev. Mod. Phys.}\ }\textbf {\bibinfo {volume} {70}},\ \bibinfo {pages} {653} (\bibinfo {year} {1998})}\BibitemShut {NoStop}%
\bibitem [{\citenamefont {Nelson}\ \emph {et~al.}(2025{\natexlab{b}})\citenamefont {Nelson}, \citenamefont {Rajakumar},\ and\ \citenamefont {Gullans}}]{JonJoel}%
  \BibitemOpen
  \bibfield  {author} {\bibinfo {author} {\bibfnamefont {J.}~\bibnamefont {Nelson}}, \bibinfo {author} {\bibfnamefont {J.}~\bibnamefont {Rajakumar}}, \ and\ \bibinfo {author} {\bibfnamefont {M.~J.}\ \bibnamefont {Gullans}},\ }\href@noop {} {\enquote {\bibinfo {title} {Limitations of noisy geometrically local quantum circuits},}\ } (\bibinfo {year} {2025}{\natexlab{b}}),\ \Eprint {http://arxiv.org/abs/2510.xxxxx} {arXiv:2510.xxxxx [quant-ph]} \BibitemShut {NoStop}%
\bibitem [{\citenamefont {Aharonov}\ and\ \citenamefont {Ben-Or}(1999)}]{aharonov1999faulttolerantquantumcomputationconstant}%
  \BibitemOpen
  \bibfield  {author} {\bibinfo {author} {\bibfnamefont {D.}~\bibnamefont {Aharonov}}\ and\ \bibinfo {author} {\bibfnamefont {M.}~\bibnamefont {Ben-Or}},\ }\href {https://arxiv.org/abs/quant-ph/9906129} {\enquote {\bibinfo {title} {Fault-tolerant quantum computation with constant error rate},}\ } (\bibinfo {year} {1999}),\ \Eprint {http://arxiv.org/abs/quant-ph/9906129} {arXiv:quant-ph/9906129 [quant-ph]} \BibitemShut {NoStop}%
\bibitem [{\citenamefont {Liu}\ \emph {et~al.}(2025)\citenamefont {Liu}, \citenamefont {Du}, \citenamefont {Liu},\ and\ \citenamefont {Ma}}]{liu2025approximatequantumerrorcorrection}%
  \BibitemOpen
  \bibfield  {author} {\bibinfo {author} {\bibfnamefont {G.}~\bibnamefont {Liu}}, \bibinfo {author} {\bibfnamefont {Z.}~\bibnamefont {Du}}, \bibinfo {author} {\bibfnamefont {Z.-W.}\ \bibnamefont {Liu}}, \ and\ \bibinfo {author} {\bibfnamefont {X.}~\bibnamefont {Ma}},\ }\href {https://arxiv.org/abs/2503.17759} {\enquote {\bibinfo {title} {Approximate quantum error correction with 1d log-depth circuits},}\ } (\bibinfo {year} {2025}),\ \Eprint {http://arxiv.org/abs/2503.17759} {arXiv:2503.17759 [quant-ph]} \BibitemShut {NoStop}%
\bibitem [{\citenamefont {Rajakumar}\ \emph {et~al.}(2025)\citenamefont {Rajakumar}, \citenamefont {Watson},\ and\ \citenamefont {Liu}}]{Rajakumar_2025}%
  \BibitemOpen
  \bibfield  {author} {\bibinfo {author} {\bibfnamefont {J.}~\bibnamefont {Rajakumar}}, \bibinfo {author} {\bibfnamefont {J.~D.}\ \bibnamefont {Watson}}, \ and\ \bibinfo {author} {\bibfnamefont {Y.-K.}\ \bibnamefont {Liu}},\ }\enquote {\bibinfo {title} {Polynomial-time classical simulation of noisy iqp circuits with constant depth},}\ in\ \href {\doibase 10.1137/1.9781611978322.30} {\emph {\bibinfo {booktitle} {Proceedings of the 2025 Annual ACM-SIAM Symposium on Discrete Algorithms (SODA)}}}\ (\bibinfo  {publisher} {Society for Industrial and Applied Mathematics},\ \bibinfo {year} {2025})\ p.\ \bibinfo {pages} {1037–1056}\BibitemShut {NoStop}%
\bibitem [{\citenamefont {Nelson}\ \emph {et~al.}(2024)\citenamefont {Nelson}, \citenamefont {Rajakumar}, \citenamefont {Hangleiter},\ and\ \citenamefont {Gullans}}]{nelson2024polynomialtimeclassicalsimulationnoisy}%
  \BibitemOpen
  \bibfield  {author} {\bibinfo {author} {\bibfnamefont {J.}~\bibnamefont {Nelson}}, \bibinfo {author} {\bibfnamefont {J.}~\bibnamefont {Rajakumar}}, \bibinfo {author} {\bibfnamefont {D.}~\bibnamefont {Hangleiter}}, \ and\ \bibinfo {author} {\bibfnamefont {M.~J.}\ \bibnamefont {Gullans}},\ }\href {https://arxiv.org/abs/2411.02535} {\enquote {\bibinfo {title} {Polynomial-time classical simulation of noisy circuits with naturally fault-tolerant gates},}\ } (\bibinfo {year} {2024}),\ \Eprint {http://arxiv.org/abs/2411.02535} {arXiv:2411.02535 [quant-ph]} \BibitemShut {NoStop}%
\bibitem [{\citenamefont {Fujii}\ and\ \citenamefont {Tamate}(2016)}]{fujii_computational_2016}%
  \BibitemOpen
  \bibfield  {author} {\bibinfo {author} {\bibfnamefont {K.}~\bibnamefont {Fujii}}\ and\ \bibinfo {author} {\bibfnamefont {S.}~\bibnamefont {Tamate}},\ }\href {\doibase 10.1038/srep25598} {\bibfield  {journal} {\bibinfo  {journal} {Sci Rep}\ }\textbf {\bibinfo {volume} {6}} (\bibinfo {year} {2016}),\ 10.1038/srep25598},\ \Eprint {http://arxiv.org/abs/1406.6932} {arXiv:1406.6932} \BibitemShut {NoStop}%
\bibitem [{\citenamefont {Ware}\ \emph {et~al.}(2023)\citenamefont {Ware}, \citenamefont {Deshpande}, \citenamefont {Hangleiter}, \citenamefont {Niroula}, \citenamefont {Fefferman}, \citenamefont {Gorshkov},\ and\ \citenamefont {Gullans}}]{ware2023sharpphasetransitionlinear}%
  \BibitemOpen
  \bibfield  {author} {\bibinfo {author} {\bibfnamefont {B.}~\bibnamefont {Ware}}, \bibinfo {author} {\bibfnamefont {A.}~\bibnamefont {Deshpande}}, \bibinfo {author} {\bibfnamefont {D.}~\bibnamefont {Hangleiter}}, \bibinfo {author} {\bibfnamefont {P.}~\bibnamefont {Niroula}}, \bibinfo {author} {\bibfnamefont {B.}~\bibnamefont {Fefferman}}, \bibinfo {author} {\bibfnamefont {A.~V.}\ \bibnamefont {Gorshkov}}, \ and\ \bibinfo {author} {\bibfnamefont {M.~J.}\ \bibnamefont {Gullans}},\ }\href {https://arxiv.org/abs/2305.04954} {\enquote {\bibinfo {title} {A sharp phase transition in linear cross-entropy benchmarking},}\ } (\bibinfo {year} {2023}),\ \Eprint {http://arxiv.org/abs/2305.04954} {arXiv:2305.04954 [quant-ph]} \BibitemShut {NoStop}%
\bibitem [{\citenamefont {Dalzell}\ \emph {et~al.}(2024)\citenamefont {Dalzell}, \citenamefont {{Hunter-Jones}},\ and\ \citenamefont {Brand{\~a}o}}]{dalzell_random_2024}%
  \BibitemOpen
  \bibfield  {author} {\bibinfo {author} {\bibfnamefont {A.~M.}\ \bibnamefont {Dalzell}}, \bibinfo {author} {\bibfnamefont {N.}~\bibnamefont {{Hunter-Jones}}}, \ and\ \bibinfo {author} {\bibfnamefont {F.~G. S.~L.}\ \bibnamefont {Brand{\~a}o}},\ }\href {\doibase 10.1007/s00220-024-04958-z} {\bibfield  {journal} {\bibinfo  {journal} {Commun. Math. Phys.}\ }\textbf {\bibinfo {volume} {405}},\ \bibinfo {pages} {78} (\bibinfo {year} {2024})},\ \Eprint {http://arxiv.org/abs/2111.14907} {arXiv:2111.14907} \BibitemShut {NoStop}%
\bibitem [{\citenamefont {Wilde}()}]{markwilde2}%
  \BibitemOpen
  \bibfield  {author} {\bibinfo {author} {\bibfnamefont {M.}~\bibnamefont {Wilde}},\ }\href@noop {} {\enquote {\bibinfo {title} {markwilde.com},}\ }\bibinfo {howpublished} {\url{http://www.markwilde.com/teaching/2015-fall-qit/lectures/lecture-19.pdf}},\ \bibinfo {note} {[Accessed 11-09-2025]}\BibitemShut {NoStop}%
\end{thebibliography}%
\clearpage
\onecolumngrid
\appendix

\section{Omitted proofs and calculations}

\subsection{Operational Interpretation of Coherent Information}
\label{app:omittedproofs}
\begin{fact}
Let $\mathcal E^{A\rightarrow B}$ be a channel from $\mathcal H_A$ to $\mathcal H_B$ where $\mathcal H_A$ has dimension $2^k$. If $I_c(\eye_A/2^k, \mathcal E) = k$ if and only if there exists a decoding channel $\mathcal D$ such that for any $\ket{\psi} \in \mathcal H_A$:
\begin{align}
    \mathcal D \circ \mathcal E(\ketbra{\psi}) = \ketbra{\psi}
\end{align}
\end{fact}
\begin{proof}
The provided proof sketch follows the notes of \cite{preskill} and adopts similar notation for consistency. As in the main text, we consider the purification of $\eye_A/2^k$ to be $\ket{\Phi}_{RA} = \sum_{i=1}^{2^k} \ketbra{i,i}{i,i}$, the channel dilation of $\mathcal E$ to be $U^{A \rightarrow BE}$, and the purification of the output to be $\ket{\phi_{RBE}}$.
We start by showing that $I_c = k$ implies that the reference and environment of $\ket{\phi_{RBE}}$ are decoupled.
\begin{align}
    I_c &= S(B) - S(E) = S(RE) - S(E)= k \\
    &\implies S(RE)= k + S(E) \\
    &\implies S(RE) = S(R) + S(E)
\end{align}

Therefore, the marginal state of $\ket{\phi}_{RBE}$ on $RE$ must be a product state. First let us purify the marginal states on $R$ and $E$ and denote these as $\ket{\psi}_{RB_1}$ and $\ket{\chi}_{B_2E}$ which together give the purification of the marginal state on $RE$ as $\ket{\psi}_{RB_1} \otimes \ket{\chi}_{B_2E}$. Next, using the fact that all purifications differ by a unitary transformation on the purifying register, we have that $\ket{\phi}_{RBE} = W_B(\ket{\psi}_{RB_1} \otimes \ket{\chi}_{B_2E})$ where $W_B$ is a unitary matrix that acts only on $B= B_1 \cup B_2$. Finally, using this fact once more, we have that there exists a unitary $V_{B_1}$ such that $\ket{\Phi}_{RA} = V_{B_1}\ket{\psi}_{RB_1}$. Therefore, we can let our decoding map be $\mathcal D = V_{B_1}W_B^\dagger$, which exactly recovers the maximally entangled state. Finally, we can show that recovering the purification of the maximally entangled state implies the ability to recover any logical codestate. This can be done using the relative-state method (section 3.3 of \cite{preskill}) as shown below.

First, let $\ket{\varphi}_A$ be some logical codestate. We can write this as:

\begin{align}
    \ket{\varphi}_A = \sum_i \varphi_i \ket{i}_A = \sqrt{2^n} \sum_i \varphi_i \left(\prescript{}{R}{\braket{i}{\Phi}_{RA}} \right)= \sqrt{2^n} \prescript{}{R}{\braket{\varphi^*}{\Phi}_{RA}}
\end{align}
Next, using linearity we have:
\begin{align}
    \mathcal D \circ \mathcal E(\ketbra{\varphi}_A) &= 2^n \prescript{}{R}{\bra{\varphi^*}} \mathcal (\mathrm{Id}_R \otimes D \circ \mathcal E)(\ketbra{\Phi_{RA}}) \ket{\varphi^*}_R \\
    &= 2^n \prescript{}{R}{\braket{\varphi^*}{\Phi_{RA}}} \braket{\Phi_{RA}}{\varphi^*}_R \\
    &= \ketbra{\varphi}_A
\end{align}

The converse is also true in that if the purification of the maximally entangled state can be recovered, then $I_c = k$. This follows by noting that $I_c(\rho_A, \mathcal D \circ \mathcal{E}) =  k$ since the reference and environment are necessarily decoupled if the reference system is successfully purified into the channel's output. Next, the statement follows by the quantum data processing inequality, which states that a channel cannot increase the coherent information.

\end{proof}
\begin{fact}
    Let $\mathcal E^{A\rightarrow B}$ be a channel from $\mathcal H_A$ to $\mathcal H_B$ where $\mathcal H_A$ has dimension $2^k$. If $I_c(\Id_A/2^k, \mathcal E) = -k$, then for any decoding channel $\mathcal D$ and any two codestates $\ket{\psi_1}$ and $\ket{\psi_2}$,
    \begin{align}
    \mathcal D \circ \mathcal E(\ketbra{\psi_1}) = \mathcal D \circ \mathcal E(\ketbra{\psi_2}) 
\end{align}
i.e. no logical information can be transmitted.
\end{fact}
\begin{proof}
     This is seen by observing that when $I_c = -k$ then $I(R:B) =0$ even after any decoding map is applied due to the quantum data processing inequality. Let $\rho_{RB} =  (\mathrm{Id}_R \otimes \mathcal D \circ \mathcal E)(\ketbra{\Phi_{RA}})$. Since $I(R:B) = 0$ we have that $\rho_{RB} = \rho_R \otimes \rho_B = \frac{\Id_R}{2^n} \otimes \rho_B$. Then, we have that for any logical state $\ketbra{\varphi}_A$:
\begin{align}
    \mathcal D \circ \mathcal E(\ketbra{\varphi}_A) &= 2^n \prescript{}{R}{\bra{\varphi^*}} \mathcal (\mathrm{Id}_R \otimes \mathcal D \circ \mathcal E)(\ketbra{\Phi_{RA}}) \ket{\varphi^*}_R \\
    &=2^n \prescript{}{R}{\bra{\varphi^*}}\frac{\Id_R}{2^n} \otimes \rho_B \ket{\varphi^*}_R \\
    &= \rho_B
\end{align}
Therefore, no matter what the logical state is, the output will always be the same, and so no information can be recovered.
\end{proof}

\subsection{Calculating the Log-Average Purity}
\label{app:purity}
Recall that the log-average purity is defined as:
\begin{align}
    I_p = -\log( \E_{\mathcal C} \Tr \rho_B^2) + \log( \E_{\mathcal C} \Tr\rho^2)
\end{align}

We can rewrite the right-hand side of this equation in a way that is easier to calculate by first introducing a second copy of the state giving $\rho \otimes \rho$. Next, we introduce the notation $\mathcal S$, which denotes the swap operator between two copies of a qubit, and $\mc I$, which denotes the identity operator acting on two copies of a qubit. We further define $\mathbb S := \mathcal S^{\otimes n}$ and $\mathbb S_A := \mc I^{\otimes |\bar{A}|}\otimes \mathcal S^{\otimes |A|}$. Using this notation, we have the following identity:
\begin{align}
    \Tr \rho_C^2 = \Tr \mathbb S_C (\rho \otimes \rho)
\end{align}
for an arbitrary region $C$. Applying this to \cref{eq:Ip}, we get:
\begin{align}
    I_p =  -\log( \E_{\mathcal C} \Tr \mathbb S_B \rho^{\otimes 2}) + \log( \E_{\mathcal C} \Tr \mathbb S \rho^{\otimes 2})
\end{align}

Recall that $\rho = (\mathrm{Id}_R \otimes \mathcal E)(\ketbra{\Phi_{RA}})$. Plugging this in, we can make further rearrangements as follows:
\begin{align}
    -\log( \E_{\mathcal C} \Tr \mathbb S_B \rho^{\otimes 2}) &= -\log( \E_{\mathcal C} \Tr \mathbb S_B (\mathrm{Id}_R \otimes \mathcal E)^{\otimes 2} (\ketbra{\Phi_{RA}}^{\otimes 2})) \\
    &= -\log( \Tr \mathbb S_B \E_{\mathcal C} \mathcal E^{\otimes 2} \left(\frac{\mc I}{4}\right)^{\otimes k}) 
\end{align}
And similarly for the second term:

\begin{align}
    \log( \E_{\mathcal C} \Tr \mathbb S \rho^{\otimes 2}) &= \log( \E_{\mathcal C} \Tr \mathbb S (\mathrm{Id}_R \otimes \mathcal E)^{\otimes 2} (\ketbra{\Phi_{RA}}^{\otimes 2})) \\
    &= \log(\Tr \mathbb S_B \E_{\mathcal C} \mathcal E^{\otimes 2} \left(\frac{\mathcal S}{4}\right)^{\otimes k}) \\
    &=
    \log(\Tr \mathbb S_B \E_{\mathcal C} \mathcal E^{\otimes 2} \left(\frac{\mathcal S}{2}\right)^{\otimes k}) - k
\end{align}
Notice that in the above expression we normalize $\mc I$ and $\mathcal S$ by their respective traces. We begin our calculation of $\mathcal E^{\otimes 2} (\left(\frac{\mc I}{4}\right)^{\otimes k})$ by applying two copies of the first layer of $\mathcal C$ to $\left(\frac{\mc I}{4}\right)^{\otimes k} \otimes (\ketbra{0}^{\otimes 2})^{ \otimes n-k}$. We let the first layer of gates be single-qubit gates drawn from the Haar measure. Notice that introducing this layer does not change the expectation since the Haar measure is both left and right invariant over the unitary group. Let us denote the second moment operator for the single-qubit Haar random gate acting on qubit $i$ as
\begin{align}
    \mathcal M_i (O) =  \E_{V \sim \mathcal H} (V\otimes V) O (V \otimes V)^{\dagger} \label{eq:single-qubit}
\end{align}
where $O$ is an operator on both copies of qubit $i$. The key insight is that this average, also known as the second moment operator of the Haar measure, projects any operator onto the symmetric subspace, which is spanned by tensor products of $\mathcal S/2$ and $\mc I/4$. In fact, every remaining operation leaves our state in this symmetric subspace. This is still intractable, however, since this space is exponentially large. In order to simplify the calculation further, we follow \cite{ware2023sharpphasetransitionlinear}, which considers an all-to-all circuit geometry where qubits are randomly permuted between layers of gates. We also adopt much of the notation from \cite{ware2023sharpphasetransitionlinear} for consistency. Now, after symmetrizing over permutations, our state is constrained to a subspace spanned by $n+1$ basis operators, which are defined by taking the average over all tensor products of the swap and identity operators with the same Hamming weight of total swap operators. More specifically,
we denote the basis states by $\{\ket{S}\}$ where $S \in [0,n]$ specifies the Hamming weight, 
\begin{align}
    \ket{S} = \frac{1}{{n \choose S}}\sum_{\sigma \in \{0,1\}^{n} \text{ s.t. } |\sigma| = S}\left(\frac{\mc I}{4}\right)^{1-\sigma_i}\left(\frac{\mathcal S}{2}\right)^{\sigma_i}
\end{align}

With this basis, we can formulate our problem as applying a series of transfer matrices \cite{dalzell_random_2024,ware2023sharpphasetransitionlinear}. 
Deferring the calculation details to the proof of \cref{fact:firstlayer}, we have that the state after the first layer of single-qubit gates and random permutations is the following:

\begin{align}
    \E_{P \sim S_n} P \circ \mathcal M_n \circ \cdots \circ \mathcal M_1(\left(\frac{\mc I}{4}\right)^{\otimes k} \otimes (\ketbra{0}^{\otimes 2})^{ \otimes n-k})) = \frac{1}{3^{n-k}}\sum_{S = 0}^{n-k} {n-k \choose S} 2^{n-k-S}\ket{S}
\end{align}
\begin{align}
    \E_{P \sim S_n} P \circ \mathcal M_n \circ \cdots \circ \mathcal M_1(\left(\frac{\mathcal S}{2}\right)^{\otimes k} \otimes (\ketbra{0}^{\otimes 2 })^{\otimes n-k})) = \frac{1}{3^{n-k}}\sum_{S = 0}^{n-k} {n-k \choose S} 2^{n-k-S}\ket{S+k}
\end{align}
where $S_n$ is the $n$-element permutation group. 

Next, we must show how to update this state after each layer of 2-qubit gates and noise channels. To do this, it is sufficient to show the action of each layer on each basis state $\ket{S}$. Let us denote the second moment operator for the two-qubit Haar random gate acting on qubits $i$ and $j$ as
\begin{align}
    \mathcal M_{i,j} (O) =  \E_{V \sim \mathcal H} (V\otimes V) O (V \otimes V)^{\dagger}
\end{align}
where $V$ is a two-qubit gate acting on qubits $i$ and $j$ and $O$ is an operator on both copies of qubits $i$ and $j$..

Then we define the transfer matrix $M_{S',S}$, which corresponds to the action of a layer of 2-qubit gates followed by a random permutation:
\begin{align}
\label{eq:transfermatrix}
    \E_{P \in S_n} P \circ \mathcal M_{n-1,n} \circ \cdots \circ \mathcal M_{1,2} (\ket{S}) = \sum_{S'=0}^n M_{S',S} \ket{S'}
\end{align}

It can be shown (see \cref{fact:2qubit}) that $M_{S',S}$ is defined as follows:

\begin{align}
    M_{S',S} = \frac{1}{{n\choose S}} \sum_{n_0,n_1,n_2} 2^{n_1} {n/2 \choose n_0,n_1,n_2} \delta_{S,n_1+2n_2} \left[\sum_{a,b} {n_1 \choose a,b} \left(\frac{4}{5}\right)^a\left(\frac{1}{5}\right)^b \delta_{S',2b+2n_2}\right]
\end{align}

Finally, we must give the transfer matrix associated with applying a layer of single-qubit noise channels. In this case we apply random single-qubit gates right after the noise channel, which has no physical effect due to the left-right invariance of the Haar measure. The transfer matrix can be defined as

\begin{align}
    \label{eq:noisetransfer}
    N_{S',S} = \sum_{a = \max(0,S-S')}^{\min(S,N-S')} {S \choose a} {N-S \choose S' - (S-a)}\delta^{S' - (S-a)} \gamma^{a} (1-\delta)^{N - S' -a} (1- \gamma)^{S-a}
\end{align}
where $\delta$ and $\gamma$ are parameters of the noise channel that depend on the noise rate. See \cref{fact:noise} for further details in the derivation of this transfer matrix. For the depolarizing channel defined in \cref{eq:depolarizing}, $\delta = 0$ and $\gamma = 1-(1-\frac{4p}{3})^2$. Setting $\delta = 0$, the transfer matrix in \cref{eq:noisetransfer} simplifies to: 
\begin{align}
    N_{S',S} = {S \choose S'} (1-\gamma)^{S'}\gamma^{S-S'}
\end{align}
where $N_{S',S} = 0$ if $S<S'$. For the damping channel defined in \cref{eq:damping}, $\delta = \frac{p^2}{3}$ and $\gamma = \frac{2}{3}(2p-p^2)$.

Finally, after applying the transfer matrices for $d$ layers of gates, permutations, and noise channels, we have obtained a decomposition of $\E_{\mathcal C} \mathcal E^{\otimes 2} \left(\frac{\mc I}{4}\right)^{\otimes k}$ and $\E_{\mathcal C} \mathcal E^{\otimes 2} \left(\frac{\mc S}{2}\right)^{\otimes k}$ in basis vectors spanned by $\ket{S}$. We can then efficiently compute $\Tr \mathbb S\E_{\mathcal C} \mathcal E^{\otimes 2} \left(\frac{\mc I}{4}\right)^{\otimes k}$ and  $\Tr \mathbb S\E_{\mathcal C} \mathcal E^{\otimes 2} \left(\frac{\mc S}{2}\right)^{\otimes k}$ using the fact that
\begin{align}
    \Tr\mathbb{S} \ket{S} = 2^{2S-n}
\end{align}
and linearity of trace.

\begin{fact}
    \label{fact:firstlayer}
\begin{align}
    \E_{P \sim S_n} P \circ \mathcal M_n \circ \cdots \circ \mathcal M_1(\left(\frac{\mathcal I}{4}\right)^{\otimes k} \otimes (\ketbra{0}^{\otimes 2})^{ \otimes n-k})) = \frac{1}{3^{n-k}}\sum_{S = 0}^{n-k} {n-k \choose S} 2^{n-k-S}\ket{S}
\end{align}
\begin{align}
    \E_{P \sim S_n} P \circ \mathcal M_n \circ \cdots \circ \mathcal M_1(\left(\frac{\mathcal S}{2}\right)^{\otimes k} \otimes (\ketbra{0}^{\otimes 2 })^{\otimes n-k})) = \frac{1}{3^{n-k}}\sum_{S = 0}^{n-k} {n-k \choose S} 2^{n-k-S}\ket{S+k}
\end{align}
where $S_n$ is the $n$-element permutation group and $\mathcal M_i$ is defined as in \cref{eq:single-qubit}.
\end{fact}
\begin{proof}
The Haar average of an operator $O$ is defined as follows:

\begin{align}
\label{eq:haaravg}
    \E_{V \sim \mathcal H} (V\otimes V) O (V \otimes V)^{\dagger} &= a \mc I + b \mathcal S,
\end{align}
where 
\begin{align}
    a = \frac{\Tr O - \frac{1}{2}\Tr O\mathcal S}{3}
\end{align}
and 
\begin{align}
    b = \frac{\Tr OS - \frac{1}{2}\Tr O}{3}
\end{align}

The only relevant input operators are $\ketbra{0}^{\otimes 2}$, $\mathcal S$, $\mc I$. We have
\begin{align}
     \E_{V \sim \mathcal H} (V\otimes V) \ketbra{0}^{\otimes 2} (V \otimes V)^{\dagger} &= \mc I/6 + \mathcal S/6,
\end{align}
\begin{align}
     \E_{V \sim \mathcal H} (V\otimes V) \mathcal S (V \otimes V)^{\dagger} &= \mathcal S,
\end{align}
\begin{align}
     \E_{V \sim \mathcal H} (V\otimes V) \mathcal I (V \otimes V)^{\dagger} &= \mc I,
\end{align}

Now averaging this uniformly over qubit permutations we get:
\begin{align}
    \frac{1}{n!} \sum_{P \in S_n} P \circ M_n \circ \cdots \circ \mathcal M_1(\left(\frac{\mathcal I}{4}\right)^{\otimes k} \otimes \ketbra{0}^{\otimes n-k}) = \frac{1}{3^{n-k}}\sum_{S = 0}^{n-k} {n-k \choose S} 2^{n-k-S}\ket{S}
\end{align}
where $S_n$ is the $n$-element permutation group. A similar calculation shows the following as well:
\begin{align}
    \frac{1}{n!} \sum_{P \in S_n} P \circ M_n \circ \cdots \circ \mathcal M_1(\left(\frac{\mathcal S}{2}\right)^{\otimes k} \otimes \ketbra{0}^{\otimes n-k}) = \frac{1}{3^{n-k}}\sum_{S = 0}^{n-k} {n-k \choose S} 2^{n-k-S}\ket{S+k}
\end{align}
\end{proof}
\begin{fact}
\label{fact:2qubit}
    \begin{align}
    M_{S',S} = \frac{1}{{n\choose S}} \sum_{n_0,n_1,n_2} 2^{n_1} {n/2 \choose n_0,n_1,n_2} \delta_{S,n_1+2n_2} \left[\sum_{a,b} {n_1 \choose a,b} \left(\frac{4}{5}\right)^a\left(\frac{1}{5}\right)^b \delta_{S',2b+2n_2}\right]
\end{align}
where $M$ is defined as in \cref{eq:transfermatrix}
\end{fact}
\begin{proof}[Proof sketch]
    It is helpful to first note the 2-qubit transfer matrix $M$ is given by
    \begin{align}
        M = 
\begin{bmatrix}
1 & 4/5 & 4/5 & 0 \\
0 & 0 & 0 & 0 \\
0 & 0 & 0 & 0 \\
0 & 1/5 & 1/5 & 1
\end{bmatrix}
\end{align} where the columns and rows are indexed by $\mc I \otimes \mc I$, $\mc I \otimes \mathcal S$, $\mathcal S \otimes \mc I$, and $\mathcal S \otimes \mathcal S$. This can be computed using \cref{eq:haaravg}. In the above definition of $M$, $n_0$, $n_1$, $n_2$ represent the number of gates which act on a tensor product of $0$, $1$, or $2$ swap operators. Among the gates that encounter $1$ swap operator (there are $n_1$ such gates), $a$ and $b$ represent the number of these gates that output a tensor product of $0$ or $2$ swap operators.
\end{proof} 
\begin{fact}
    \label{fact:noise}
    \begin{align}
    \label{eq:noisetransferapp}
    N_{S',S} = \sum_{a = \max(0,S-S')}^{\min(S,N-S')} {S \choose a} {N-S \choose S' - (S-a)}\delta^{S' - (S-a)} \gamma^{a} (1-\delta)^{N - S' -a} (1- \gamma)^{S-a}
\end{align}
\end{fact}
\begin{proof}[Proof sketch]
The single-qubit transfer matrix of a noise channel followed by a random single-qubit gate takes the general form of 
\begin{align}
    N = \begin{bmatrix}
        1 - \delta & \gamma \\
        \delta & 1-\gamma
    \end{bmatrix}
\end{align}
where the rows and columns are indexed by $\mc I$ and $\mathcal S$. $\delta$ and $\gamma$ can both be calculated by applying \cref{eq:haaravg} for the given noise channel. For a given qubit, there are four possibilities allowed by this transfer matrix: $\mc I \rightarrow \mc I$, $\mc I \rightarrow \mathcal S$, $\mathcal S \rightarrow \mc I$, and $\mathcal S \rightarrow \mathcal S$. In \cref{eq:noisetransferapp}, $a$ represents the number of qubits that undergo the transition $\mathcal S \rightarrow \mc I$. Notice that if $S > S'$ then there must be at least $S - S'$ qubits that experience $\mathcal S \rightarrow \mc I$ which is why $a$ takes this as its minimum value. In addition, the output contains at most $N-S'$ $\mc I$ operators. Furthermore, there are only $S$ qubits that contain $\mathcal S$ operator. Therefore $a$ cannot exceed $N-S'$ or $S$ which is why these form the upper limits of the sum. Now, the remaining quantities can be deduced from $a$ as follows. There are $S-a$ qubits experiencing the transition $\mathcal S \rightarrow \mathcal S$, $S'-(S-a)$ experiencing the transition $\mc I \rightarrow \mathcal S$, and $N-S'-a$ qubits experiencing the transition $\mc I \rightarrow \mc I$. Accounting for the degeneracies and including the appropriate factor for each transition type results in the desired equation.
\end{proof}

\subsection{Proof of \Cref{theorem:worst-case}}
\label{app:worstcaseproof}
Our proof strategy is to relate coherent information to a \textit{relative entropy distance} which is defined as $D(\rho||\zeta) := \Tr\rho(\log\rho - \log \zeta)$, where we can apply the following convergence result for noisy quantum circuits,
\begin{fact}[Follows from Lemma 3.2 of \cite{JonJoel}]
    Let $X$ and $Y$ be quantum systems, and let $\rho_{XY}$ denote the joint state obtained as follows: start from any input state on $XY$, apply an arbitrary quantum circuit acting only on $X$, which has $d$ layers of interspersed depolarizing noise of strength $p$ acting independently on each qubit in $X$. Then,
    \begin{align}
        D(\rho_{XY}\|\frac{\eye_X}{2^{|X|}} \otimes \rho_{Y}) \leq (1-p)^d|XY|
    \end{align}
\end{fact}

\begin{theorem}[Restatement of \Cref{theorem:worst-case}]
    Let $\tilde{\mc C}$ be an arbitrary quantum circuit acting on $n$ qubits, which has $d$ layers of interspersed depolarizing noise of strength $p$ acting independently on each qubit. Then the coherent information of the encoding channel which encodes a $k$-qubit quantum state $\zeta$ as $\mc E(\zeta) = \tilde{C}(\zeta \otimes \ketbra{0}^{n-k})$, is bounded as follows:
    \begin{align}
        I_c(\frac{\eye}{2^k}, \mc E) \leq e^{-p d}(n+k)-k
    \end{align}
\end{theorem}
\begin{proof}
    We will consider the Hilbert spaces corrresponding to $A,B,R,E$ and entropies of the reduced density matrices of $\ketbra{\phi}_{RBE}$, as defined in \cref{sec:background}. We will denote $\rho = \Tr_E(\ketbra{\phi}_{RBE}) = (\mathrm{Id}_R \otimes \mc E)(\ketbra{\Phi_{RA}}) $. We also use the known identity that $S(B|R) = |B|-D(\rho_{RB}\| \rho_R \otimes \frac{\eye_B}{2^{|B|}})$ \cite{markwilde2}.
    \begin{align*}
        I_c &= S(B) - S(E)\\
        &= S(B)- S(RB)\\
        &= S(B)-S(B|R) - S(R)\\
        &\leq n - S(B|R)-k \tag{since $\rho_R$ is maximally mixed}\\
        &= n - (n-D(\rho_{RB}\| \rho_R \otimes \frac{\eye_B}{2^{|B|}})) - k \\
        &= D(\rho_{RB}\| \rho_R \otimes \frac{\eye_B}{2^{|B|}}) - k \\
        &= (1-p)^d (n+k) - k \\
        &\leq e^{-pd}(n+k) - k
    \end{align*}
\end{proof}

\section{Supplementary Plots}
\label{app:plots}
\begin{figure}[h]
    \centering
    \includegraphics[width=0.6\textwidth]{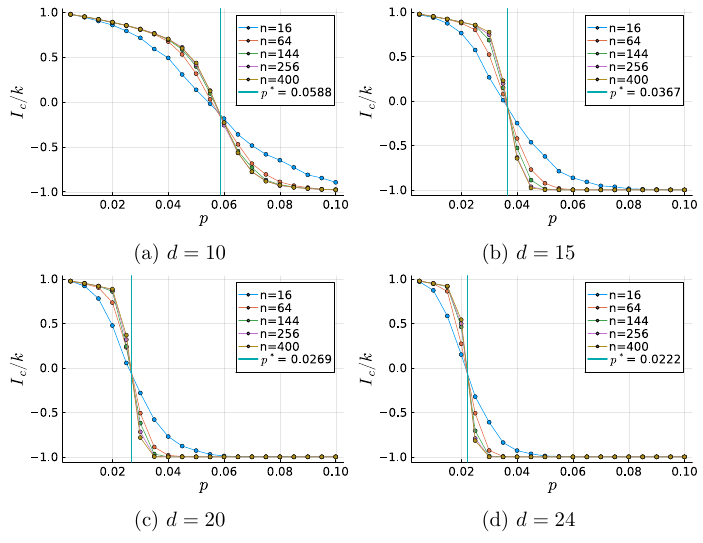}
    \caption{$I_c/k$ vs $p$ for 2D random Clifford circuits at various depths $d=10,\dots,24$ where many depth values are omitted to conserve space. The estimated critical error rates for each of these plots are used to generate \cref{fig:clifford}.} 
\end{figure}
\begin{figure}[h]
    \centering
    \includegraphics[width=0.6\textwidth]{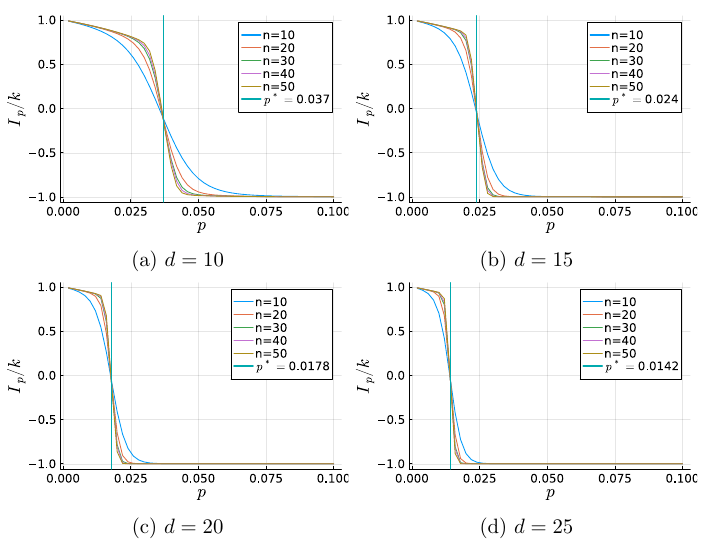}
    \caption{$I_p/k$ vs $p$ for Haar random circuits with depolarizing noise at various depths $d=5,\dots,25$ where many depth values are omitted to conserve space. The estimated critical error rates for each of these plots are used to generate \cref{fig:haardepol}.} 
\end{figure}
\begin{figure}[h]
    \centering
    \includegraphics[width=0.6\textwidth]{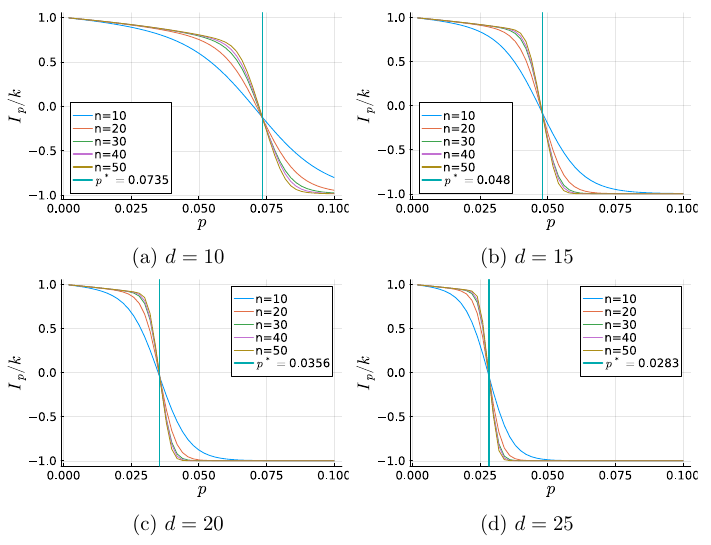}
    \caption{$I_p/k$ vs $p$ for Haar random circuits with amplitude damping noise at various depths $d=10,\dots,35$ where many depth values are omitted to conserve space. The estimated critical error rates for each of these plots are used to generate \cref{fig:haaramp}.} 
\end{figure}

\end{document}